\documentclass[11pt]{article}
\usepackage{graphicx} 

\usepackage{fancyhdr}
\usepackage{isomath}
\usepackage{mathtools} 
\usepackage{amsbsy}
\usepackage{amssymb,amsmath,amsthm}
\usepackage{amscd,amsfonts}
\usepackage{bigints}
\usepackage{graphicx}
\usepackage{verbatim}
\usepackage{euscript}
\usepackage{alltt}
\usepackage{stmaryrd}
\usepackage{relsize}
\usepackage{enumerate}
\usepackage{url}
\usepackage[stable]{footmisc}
\usepackage{breakurl}
\usepackage{hyperref}
\usepackage{comment}

\usepackage[maxbibnames=9,sorting=none,style=numeric]{biblatex}
\usepackage[font=small,labelfont=bf]{caption}
\usepackage[font=small,labelfont=bf]{subcaption}
\usepackage{float}
\usepackage{mwe}
\usepackage{algorithm}
\usepackage{algpseudocode}
\usepackage{xcolor}
\usepackage[super]{nth}
\usepackage{soul}

\DeclareGraphicsExtensions{.eps,.pdf}
\usepackage{amsmath}
\usepackage{amsbsy}
\usepackage{amssymb}
\usepackage{amscd}
\usepackage{amsfonts}
\usepackage{isomath}

\newcommand{\R}{\mathbb R}
\newcommand{\N}{\mathbb N}

\newcommand{\K}{\mathbb{K}}

\newcommand{\eps}{{\epsilon}}

\newcommand{\beq}{\begin{equation}}
\newcommand{\eeq}{\end{equation}}
\newcommand{\beqs}{\begin{eqnarray}}
\newcommand{\eeqs}{\end{eqnarray}}
\newcommand{\beql}{\begin{equation} \label}
\newcommand{\half}{\frac{1}{2}}

\newtheorem{theorem}{Theorem}[section]
\newtheorem{lemma}{Lemma}[section]
\newtheorem{proposition}{Proposition}[section]
\newtheorem{remark}{Remark}[section]

\newcommand{\curl}{\mathop{\rm curl}\nolimits}

\newcommand{\p}{\partial}

\newcommand{\mca}{\mathcal{A}}

\newcommand{\dee}{\mathcal{D}}
\newcommand{\scl}{\mathcal{L}}

\newcommand{\la}{{\langle}}
\newcommand{\ra}{{\rangle}}

\newcommand{\Tr}{{\rm {Tr}}}

\newcommand{\gog}{{\mathfrak{G}}}

\newtheorem{definition}{Definition}[section]

\usepackage[margin=1in]{geometry}

\date{}
\begin{document}
\title{Variational Dual Solutions of Chern-Simons Theory}

\author{Amit Acharya\thanks{Department of Civil \& Environmental Engineering, and Center for Nonlinear Analysis, Carnegie Mellon University, Pittsburgh, PA 15213, email: acharyaamit@cmu.edu.} $\qquad$ Janusz Ginster\thanks{Weierstrass Institute, Mohrenstrasse 39, 10117 Berlin, Germany, email: ginster@wias-berlin.de.} $\qquad$ Ambar N. Sengupta\thanks{Department of Mathematics, University of Connecticut, Storrs, CT 06269, email: ambarnsg@gmail.com.}}

\maketitle
\begin{abstract}
\noindent A scheme for generating weakly lower semi-continuous action functionals corresponding to the Euler-Lagrange equations of Chern-Simons theory is described.  Coercivity is deduced for such a functional in appropriate function spaces to prove the existence of a minimizer, which constitutes a solution to the Euler-Lagrange equations of Chern-Simons theory in a relaxed sense. A geometric analysis is also made, especially for the gauge group $SU(2)$, relating connection forms on the bundle to corresponding forms in the dual scheme.

\end{abstract}

\section{Introduction}
The goal of this paper is to define a dual variational principle corresponding to the Euler-Lagrange equations of the Chern-Simons functional and to prove existence of minimizers of the dual functional. The functional has the property that its critical points formally define solutions to the Euler-Lagrange equations of Chern-Simons theory \cite{ChernSim1974}. Our approach is an application of a recently developed scheme for generating variational principles corresponding to a given set of equations and side conditions \cite{Acharya:2022:VPN,Acharya:2023:DVP,AcharyaGinsterSingh,Acharya:2024:HCC,Acharya:2024:APD,Acharya:2024:VDS,Kouskiya:2024:HCH,KA2}; the equations can represent dissipative or conservative systems. The scheme treats the primal PDE under consideration as constraints and invokes a more-or-less arbitrarily designable strictly convex, auxiliary potential to be optimized. Then, a \emph{convex} dual variational principle
for the Lagrange multiplier (dual) fields emerges involving a dual-to-primal (DtP) mapping
(an adapted change of variables), with the special property that its Euler-Lagrange equations
are exactly the primal PDE system, interpreted as equations governing the dual fields using the DtP
mapping. These ideas are closely related to those of Brenier \cite{Brenier:2010:HCN,Brenier:2018:IVP} on `hidden convexity' of nonlinear PDE.

The Chern-Simons functional is not bounded above or below and hence neither minimizers nor maximizers exist (see the discussion in Section \ref{sec:var_anal}). This renders the Direct Method of the Calculus of Variations not feasible to prove existence of solutions to the Euler-Lagrange equations. Instead, the proposed method in our approach associates the original problem to the Euler-Lagrange equations of a \textit{dual} functional. We prove that for specific choices of auxiliary potentials this dual functional turns out to be bounded below, lower-semicontinuous and coercive. In particular, the Direct Method of the Calculus of Variations can be invoked to show existence of minimizers of the dual functional. Assuming that the dual functional is regular enough and the DtP mapping is well-defined, these minimizers correspond to extremals of the Chern-Simons theory. This is why we refer to them as \emph{variational dual solutions} of the theory.

An outline of the paper is as follows: In Sec.~\ref{sec:geom_back} we  present a largely self-contained account of the Chern-Simons functional and the associated variational equations (\ref{E:FA0}), describing flat connections, which we spell out more explicitly (\ref{E:flatSU2}) for the case where the gauge group is $SU(2)$; following this,  in section \ref{ss:gdds}, we work out the geometry of the dual scheme for the  variational equations. In Theorem \ref{T:SHcrit}, we show that the critical points of the dual functional $S_H$ correspond to flat connections with the prescribed boundary values. The explicit expression for the dual functional 
for a quadratic auxiliary potential is developed in Sec.~\ref{sec:quad_H} and the consistency of the scheme is shown, i.e., that there exists at least one dual functional whose critical point delivers a solution to the Euler-Lagrange equations of Chern-Simons theory satisfying the prescribed boundary conditions, for every solution of the theory. In Sec.~\ref{sec:var_anal} we undertake a variational analysis of dual functionals corresponding to Chern-Simons theory, and prove the existence of variational dual solutions of the same (Theorem \ref{thm: existence}).

\section{Geometric background}\label{sec:geom_back}

One of the foundational results in geometry is the Gauss-Bonnet theorem, which says that the integral of the curvature over a closed surface is $2\pi$ times the Euler characteristic, a topological invariant, of the surface. A far-reaching generalization was proved by Chern \cite{Chern1944} a hundred years later, showing that the integral of the Pfaffian (a polynomial in the matrix entries) of the curvature 2-form over  a closed orientable $(2n)$-dimensional manifold equals $(2\pi)^n$ times the Euler characteristic of the manifold. In proving the Gauss-Bonnet theorem and Chern's theorem, it is useful to construct a potential $\Phi$ for the Pfaffian of the curvature. These ideas generalize to the setting of connections on principal bundles: the Pfaffian of the curvature $2$-form, corresponding to a connection $1$-form $\omega$ on a principal $G$-bundle $\pi:P\to M$, is a closed  $(2n)$-form that descends to a closed $(2n)$-form ${\rm Pf}(\Omega)$ on $M$.   The $(2n)$-form on $P$ has a global potential $\Phi$ that descends locally, but, in general, not globally, to a $(2n-1)$-form on the base manifold $M$ (this is a special case of (\cite[Proposition 3.2]{ChernSim1974}). In the case of a 4-dimensional manifold, this potential is the Chern-Simons form, with a suitable normalization constant. If a closed orientable manifold were to be split into two 4-dimensional contractible submanifolds with a common 3-dimensional boundary, then the integral of ${\rm Pf}(\Omega)$ over each submanifold equals, by Stokes' theorem, the integral of the local potential over the boundary $3$-manifold. Our focus is on the Chern-Simons $3$-form for the group $SU(2)$ over an open subset of $\R^3$.

\subsection{The Chern-Simons form} We work with a smooth $1$-form $A$, defined on a compact oriented $3$-manifold $M$, possibly with boundary, with values in a Lie algebra $\gog$, which is equipped with a symmetric, bilinear pairing $\la\cdot,\cdot\ra$ on $\gog\times\gog$, satisfying the ad-invariance condition (\ref{E:uvw{}}). We define the {
\em Chern-Simons} $3    $-form to be:
\begin{equation}\label{E:csA}
    {\rm cs}(A)= \left\langle A\wedge dA+\frac{1}{3}A\wedge [A\wedge A]\right\rangle,
\end{equation}
with $\la\cdot\ra$ being as defined in (\ref{E:{AdotB]}}). 
We take the Chern-Simons action to be the integral of this form over  $\R^3$ :
\begin{equation}\label{E:CSAOmega}
{\rm CS}(A)=\int_{M} {\rm cs}(A).
\end{equation}
(As remarked earlier,  a certain normalization constant multiplier is needed in other contexts.)

An example of interest is when 
$$M=\Sigma\times [0,1],$$
where $\Sigma$ is a compact, oriented surface. We can consider this as embedded as a submanifold-with-boundary in $\R^3$, by considering the surface $\Sigma$ to be sitting inside $\R^3$. A different example, is given by simply a closed ball in $\R^3$.

\subsubsection{The case of the gauge group $SU(2)$}
Let us specialize to the case of $SU(2)$, with Lie algebra $su(2)$. Notation and basic facts about $SU(2)$ and $su(2)$ are summarized in section \ref{ss:su2}. We consider    $su(2)$-valued smooth $1$-forms $A$ on a $3$-manifold, and we work locally with a coordinate chart $(x^1, x^2, x^3)$. (The form $A$ is  to be thought of as a connection form on a bundle pulled down by means of a section.) The {Chern-Simons} $3$-form for    $A$ is then
\begin{equation}\label{E:csAmatr}
    {\rm cs}(A)=-  {\rm Tr}\left(A\wedge dA+\frac{2}{3}A\wedge A\wedge A\right).
\end{equation}
We work with a fixed orthonormal basis  $\{E_1, E_2, E_3\}$ of $su(2)$, and the standard basis $1$-forms $\{dx^1, dx^2, dx^3\}$ for the space of ordinary real-valued $1$-forms. Then we can write an $su(2)$-valued $1$-form $A$ as
\begin{equation}
    A= A_p \otimes dx^p=\sum_{J=1}^3A_{Jp}E_J {\otimes} dx^p,
\end{equation}
where each 
$$A_p=\sum_{J=1}^3A_{Jp}E_J $$
is an $su(2)$-valued function on $\Omega$, with each  $A_{JP}$ being a real-valued function on $\Omega$. Then the wedge product $A\wedge dA$ is 
\begin{equation}
    A\wedge dA =\sum_{p,q,r} A_p\partial_q A_r dx^p\wedge dx^q\wedge dx^r
\end{equation}
and
\begin{equation}
    A\wedge A\wedge A =\sum_{p,q,r} A_pA_qA_r dx^p\wedge dx^q\wedge dx^r.
\end{equation}
Using skew-symmetry of the wedge product it can be checked that these are both $su(2)$-valued. 
 Then 
\begin{equation}\label{E:CSformcoords}
    \begin{split}
     {\rm cs}(A) 
      &= -{\rm Tr} \left( A\wedge  dA+\frac{2}{3} A\wedge A\wedge A \right) \\
      &= -{\rm Tr} \left(A_p\partial_qA_r+ \frac{2}{3}A_pA_qA_r\right)\,dx^p\wedge dx^q\wedge dx^r\\
       &= -{\rm Tr} \left(A_p\partial_qA_r+ \frac{1}{3}A_p[A_q,A_r] \right)\,dx^p\wedge dx^q\wedge dx^r\\
       &= -{\rm Tr} \left(A_{Jp}\partial_qA_{Kr}(E_JE_K)+ \frac{1}{3}A_{Jp}A_{Kq}A_{Lr}E_J[E_K, E_L] \right)\,dx^p\wedge dx^q\wedge dx^r\\
       &=   2 \varepsilon_{pqr} \left( A_{Jp} \partial_q A_{Jr}  + \frac{2}{3} \varepsilon_{JKL} A_{Jp}A_{Kq} A_{Lr} \right) \,dx^1\wedge dx^2\wedge dx^3,
    \end{split}
\end{equation}
where we have used the commutation relations (\ref{E:EJKLcomm}), the orthonormality (\ref{E:TrEJK}), and (\ref{E:TrEJKL}).

\subsection{The Euler-Lagrange equations for the  Chern-Simons action}

To compute the directional derivative of the CS action, for  any   smooth $1$-form $v$ on $M$, a compact oriented $3$-manifold with boundary, we have
\begin{equation}
    \begin{split}
       & \frac{d}{dt}\Big|_{t=0} {\rm CS}(A+tv)\\
       &= \int_{M}\left\langle v\wedge dA+ A\wedge dv   +\frac{1}{3}\left(v\wedge [A\wedge A]+ A\wedge [v\wedge A] +A\wedge [A\wedge v]\right)\right\rangle\\
        &= \int_{M}\left\langle v\wedge dA +dv\wedge A  +\frac{1}{3}\left(v\wedge [A\wedge A]+ A\wedge [v\wedge A] +A\wedge [A\wedge v]\right)\right\rangle \\
        &\hskip 1in \hbox{(using (\ref{E:AwedgeBsym}))}\\
        &= \int_{M}\left\langle d(v\wedge  A)+2v\wedge dA  +\frac{1}{3}\left(v\wedge [A\wedge A]+ v\wedge [A\wedge A] +v\wedge [A\wedge A]\right)\right\rangle \\
        &\hskip 1in \hbox{(using  (\ref{E:ABwC}))}\\
        &=\int_{\partial M}\la v\wedge A\ra \,+ 2\int_M\left\langle v\wedge\left(dA +\frac{1}{2}[A\wedge A]\right)\right\rangle,
    \end{split}
\end{equation}
where $F^A$ is the curvature of $A$:
\begin{equation}
    F^A:=dA+\frac{1}{2}[A\wedge A].
\end{equation}
Thus,
\begin{equation}\label{E:dCSAv}
d{\rm CS}|_Av  = \int_{\partial M}\la v\wedge A\ra \,+ 2\int_M\left\langle v\wedge F^A\right\rangle.
\end{equation}
Taking $v|_{\partial M}=0$, which means that we fix the values of $A$ on the boundary, the condition that $A$ is a critical point for the Chern-Simons integral is that its curvature $F^A$ be $0$:
\begin{equation}\label{E:FA0}
    \hbox{Euler-Lagrange equations for ${\rm CS}$}:\quad F^A=0.
\end{equation}
Consider now the case  $M=\Sigma\times [0,1]$, where $\Sigma$ is a compact, oriented, surface. We denote by $t$ the coordinate along $[0,1]$. We take $A$ to be of the form  $A_{\Sigma\times\{t\}}$, with the coefficient of $dt$ being $0$; then $A$ being flat means that $A_{\Sigma\times\{t\}}$ is a flat connection over the surface $\Sigma\times\{t\}$ and is the same connection, viewed over $\Sigma$, for all $t$. In other words, $A$ corresponds to a path of flat connections over $\Sigma$.

\subsubsection{The Euler-Lagrange equations for the $SU(2)$ case}
Working now with $A$ being an $su(2)$-valued $1$-form on $\Omega\subset\R^3$, we have the variation computation:
\begin{subequations}\label{eq:CS_var}
\allowdisplaybreaks
\begin{align}
    &\delta {\rm CS}(A) \notag\\
    &= 2\epsilon_{pqr}\int_{\Omega} \left[ (\partial_qA_{Jr})\delta A_{Jp} - (\partial_qA_{Jp}) \delta A_{Jr} \right. \notag\\
    &\left.\hskip 1in +\frac{2}{3}\epsilon_{JKL}\left\{(\delta A_{Jp})A_{Kq}A_{Lr} +A_{Jp}(\delta A_{Kq})A_{Lr}+ A_{Jp}A_{Kq}{\delta}A_{Lr}\right\}\right] \,d{\rm vol} \notag\\
    &= 2\int_{\Omega} \left[\epsilon_{pqr}\partial_qA_{Jr} -\epsilon_{rqp}\partial_qA_{Jr}\right. \notag\\
    &\left. \hskip 1in+\frac{2}{3}\left\{\epsilon_{pqr}\epsilon_{JKL}A_{Kq}A_{Lr}+\epsilon_{qpr}\epsilon_{KJL}A_{Kq}A_{Lr}  +\epsilon_{rqp}\epsilon_{LKJ}A_{Lr}A_{Kq}\right\}  \right]\delta A_{Jp}\,d{\rm vol} \notag\\
    &= 2\int_{\Omega}\left[\epsilon_{pqr}\left(\partial_qA_{Jr}+\partial_qA_{Jr}\right) +\frac{2}{3}\epsilon_{pqr}\epsilon_{JKL} \left\{ A_{Kq}A_{Lr}+A_{Kq}A_{Lr}+A_{Kq}A_{Lr}\right\} \right]\delta A_{Jp}\, d{\rm vol} \tag{\ref{eq:CS_var}}\\
    &=4\int_{\Omega}\epsilon_{pqr}   \left[\partial_qA_{Jr} +\epsilon_{JKL}A_{Kq}A_{Lr}\right]        \delta A_{Jp}\,d{\rm vol}. \notag
\end{align}
\end{subequations}
For this variation to be zero for all variations in $A$, we have the condition
\begin{equation}\label{E:ELCS1}
   \epsilon_{pqr}\left[\partial_qA_{Jr} + \epsilon_{JKL}A_{Kq}A_{Lr}\right]=0.
\end{equation}
For each fixed $p$, this gives
\begin{equation}\label{E:flatSU2}
 \partial_qA_{Jr}-\partial_rA_{Jq}+\epsilon_{JKL}\left\{A_{Kq}A_{Lr}-A_{Kr}A_{Lq}\right\}=0.
\end{equation}
This is again the flatness condition (\ref{E:FA0}), just written in local coordinates.

\subsection{Geometric description of the dual scheme for Chern-Simons}\label{ss:gdds}

As always, we   work with an open subset $\Omega$ of $\R^3$, such that $\Omega\cup\partial\Omega$ is a compact, oriented manifold, in the natural structure inherited from $\R^3$.  Moreover, we also work with a Lie algebra $\gog$, equipped with an ad-invariant inner-product $\la\cdot,\cdot\ra$, and suitably smooth $\gog$-valued $1$-forms. 

We have seen in (\ref{E:FA0}) that the extrema of the Chern-Simons functional are given by flat connections. In this section we carry out, in geometric language, the  {\em dual method} \cite{Acharya:2022:VPN,Acharya:2024:HCC}: this (i) writes a connection $A$ (the primal variable here) as $A^{(H)}(\lambda)$, in terms of a dual variable $\lambda$, and (ii)  designs a variational principle   {involving} a functional $S_H(\lambda)$ such that the Euler-Lagrange equations for $S_H(\lambda)$ automatically produce both the flatness condition (\ref{E:FA0}) for $A$ and a given boundary condition for $A$ on $\partial\Omega$.   We accomplish this in Theorem \ref{T:SHcrit}.

\subsubsection{Lie-algebra-valued forms} At each $p\in\Omega$, we have the space
$${\gog\otimes T_p^*\Omega}$$
of linear maps from $  T_p\Omega $ with values in $\gog$. The union
$$ \gog \otimes (T^*\Omega) := \cup_{p\in\Omega}(\gog\otimes T_p^*\Omega) $$
is a bundle over $\Omega$, with projection map
\begin{equation}\label{E:TpOmg}
    \gog \otimes (T^*\Omega) \to \Omega : A_p\mapsto p,
\end{equation}
for $A_p\in {\gog\otimes T_p^*\Omega}$ and all $p\in\Omega$.   A section of the bundle (\ref{E:TpOmg}) is a $1$-form with values in $\gog$. 

\subsubsection{The auxiliary function $H$} 
We will work with a smooth  function $H$ defined on the bundle space  $ \gog\otimes (T^*\Omega)$, with the restriction of $H$ to the space $\gog\otimes T_p^*\Omega$ denoted by $H_p$. Then for any $A_p\in \gog\otimes (T_p^*\Omega)$, the map
\begin{equation}
  dH_p\Big|_{A_p}:  \gog\otimes T_p^*\Omega  \to \R ; \qquad  v\mapsto dH_p\Big|_{A_p}v=\frac{d}{dt}\Big|_{t=0}H_p(A_p+tv)
\end{equation}
is linear in $v$ and so there exists a unique element
$$\nabla H(A_p)\in\gog \otimes T_p^*\Omega$$
such that
\begin{equation}
      dH_p\Big|_{A_p}v  = \la \nabla H(A_p), v\ra\qquad \hbox{for all $v\in \gog\otimes T^*_p\Omega$.}
\end{equation}
Moreover, $\nabla H$, viewed as a $\gog$-valued function on $\gog\otimes (T^*\Omega) $, is smooth. Especially when working in coordinates, it will be simpler if we write $\nabla H(A)$ instead of $\nabla H(A_p)$.

To express the gradient in coordinates, let $E_1,\ldots, E_{\dim\gog}$ be an orthonormal basis of $\gog$; then $\{E_Z\otimes dx^q\}$ is an orthonormal  basis of $\gog\otimes T_p^*\Omega$, and
\begin{equation}
  \partial_{Zr}H(A):=dH_p\Big|_{A} (E_Z\otimes dx^q)= \la\nabla H(A), E_Z\otimes dx^q\ra.
\end{equation}
So we can write the gradient as
\begin{equation}\label{E:gradHAcrd}
   \nabla H(A) = \partial_{Zr}H(A)E_Z\otimes dx^r.
\end{equation}

  An example that we will explore further later is with $H(A_p)$ given by
$$\frac{1}{2}|\!|A_p-\bar{A}_p|\!|^2=\frac{1}{2} \sum_{j=1}^3\la A_p(e_j)-{\bar A}_p(e_j), A_p(e_j)-{\bar A}_p(e_j)\ra,$$
where $e_1, e_2, e_3$ are orthonormal vectors in $T_p\Omega$, $\la\cdot,\cdot\ra$ is the inner-product on $\gog$, and $\bar{A}_p$ is some chosen fixed element of ${\gog\otimes T_p^*\Omega}$; for this function, the gradient is given by
$$\nabla \frac{1}{2}|\!|A-\bar{A}|\!|^2= \sum_{j=1}^3 \left(A(e_j)-{\bar A}(e_j)\right)e^j,$$
where $e^1, e^2, e^3$ form the basis of $T^*_p\Omega$ dual to $e_1, e_2, e_3$.

\subsubsection{The connection space $\mca_k$}

Let
$$\mca_k$$
  be the set of all $C^k$ sections, with $k\geq 1$,  of the bundle $  \gog\otimes (T^*\Omega) \to \Omega$. Thus, $\mca_k$ is the space of $\gog$-valued $C^k$-smooth connections over $\Omega$.
Although in standard differential geometry, a connection is always $C^\infty$, we can work, for our purposes, with $C^k$ forms, where $k\geq 1$.

\subsubsection{The pre-dual functional  ${\hat S}_H$}
We introduce now the {\em pre-dual} functional 
\begin{equation}
    {\hat S}_H: \mca_k\times\mca_k\to\R,
\end{equation}
given by
\begin{equation}\label{E:predual}
\begin{split}
    {\hat S}_H[A, \lambda] &=  -\int_{\partial \Omega}\la\lambda\wedge A^{(b)} \ra+\int_{\Omega} \left(\left\langle d\lambda\wedge A+\half \lambda\wedge [A\wedge A] \right\rangle  +H(A)\,d{\rm vol}\right).
    \end{split}
\end{equation}
Here $A^{(b)}$ is a fixed $\gog$-valued $C^k$ $1$-form on $\partial\Omega$, which we are not displaying   in  ${\hat S}_H[A, \lambda]$ for the sake of  notational simplicity. (If we were to work at the level of a bundle $P$ over $\Omega$, then a connection would be a $\gog$-valued $1$-form on $P$; however, we keep things simple by using a global section of $P$ over $\Omega$, which makes it possible to bring all relevant forms and calculations down to $\Omega$.)

\begin{proposition}\label{P:derivSH} 
The functional ${\hat S}_H$ has directional derivatives everywhere. If $A, \lambda, v$ are    $\gog$-valued $C^k$ $1$-forms on $\Omega$, with $k\geq 1$, we have
\begin{equation}\label{E:dShatAdt2}
\begin{split}
       \partial_1{\hat S}_H(A,\lambda)|v &\stackrel{\rm def}{=}\frac{d}{dt}\Big|_{t=0}{\hat S}_H[A+tv,\lambda]\\ &= \int_{\Omega} \la \ast (d^A\lambda) +\nabla H(A), v\ra \,d{\rm vol},
    \end{split}
\end{equation}
where 
  $d^A\lambda$ is the covariant derivative:
\begin{equation}\label{E:dAlam}
    d^A{\lambda} =d\lambda+[A\wedge \lambda],
\end{equation}
and
$\ast$ is the Hodge dual operation and the gradient $\nabla H(A)$ is the $\gog$-valued $1$-form on $\Omega$  for which
\begin{equation}
      H'(A)v  = \la \nabla H(A), v\ra
\end{equation}
holds for all $v\in\mca_k$. Moreover, 
\begin{equation}
\begin{split}
  \partial_2{\hat S}_H(A,\lambda)|\beta &\stackrel{\rm def}{=}  \frac{d}{ds}\Big|_{s=0}{\hat S}_H(A, \lambda+s\beta)\\
  &= \int_{\partial\Omega}\la \beta\wedge (A-A^{(b)}\ra +\int_{\Omega}\la \beta\wedge F^A\ra
  \end{split}
\end{equation}
for all $\beta\in\mca_k$.
    
\end{proposition}

\begin{proof} By straightforward computation, we have
\begin{equation}\label{E:dShatAdt}
    \begin{split}
       \frac{d}{dt}\Big|_{t=0}{\hat S}_H[A+tv,\lambda] &= \int_{\Omega} \left( \left\langle d\lambda \wedge v +\half \lambda\wedge  [v\wedge A]+\half \lambda \wedge [A\wedge v]\right\rangle +H'(A)v \,d{\rm vol}\right).
    \end{split}
\end{equation}
  Next, using (\ref{E:AwedgeBsym}),(\ref{E:AbrackBskew}) and (\ref{E:ABwC}),   we have:
\begin{equation}
\begin{split}
    \left\langle d\lambda\wedge v +\half \lambda\wedge  [v\wedge A]+\half \lambda \wedge [A\wedge v]\right\rangle &=\left\langle d\lambda\wedge v +\lambda \wedge [v\wedge A] \right\rangle \\
    &=\left\langle d\lambda\wedge v +[A\wedge \lambda]\wedge v \right\rangle\\
    &=\left\langle d^A\lambda\wedge v\right\rangle,
    \end{split}
\end{equation}

We can then write (\ref{E:dShatAdt}) as
\begin{equation}\label{E:dShatAdt1}
\begin{split}
       \partial_1{\hat S}_H(A,\lambda)|v & = \int_{\Omega} \left( \left\langle d^A\lambda\wedge v\right\rangle +H'(A)v \,d{\rm vol}\right)\\
       &=\int_{\Omega} \la \ast (d^A\lambda) +\nabla H(A), v\ra \,d{\rm vol}.
    \end{split}
\end{equation}
The directional derivative with respect to $\lambda$ is:
\begin{equation}\label{E:dShatAdt3}
\begin{split}
  \partial_2{\hat S}_H(A,\lambda)|\beta &\stackrel{\rm def}{=}  \frac{d}{ds}\Big|_{s=0}{\hat S}_H(A, \lambda+s\beta)\\
  &= \int_{\partial\Omega}\la \beta\wedge (A-A^{(b)}\ra +\int_{\Omega}\la \beta\wedge F^A\ra 
  \end{split}
\end{equation}
\end{proof}

 We equip all spaces of $C^k$ differential forms, with values in $\gog$ or in $\R$, with a family of seminorms given by the suprema of derivatives of the components of the form. This makes each such space  of forms  a Fr\'echet space (Example 1.1.5 in \cite{Ham1982}).  

\begin{proposition}\label{P:SHC1} The function 
\begin{equation}
 {\hat S}_H: \mca_k\times\mca_k\to\R:  (A, \lambda)\mapsto  {\hat S}_H(A,\lambda)
 \end{equation}
is   directionally $C^1$.

\end{proposition}

\begin{proof} Combining (\ref{E:dShatAdt1}) and (\ref{E:dShatAdt2}), we have
\begin{equation}\label{E:hatSpHA}
 {\hat S}_H'(A,\lambda)| (v, \beta)= \int_{\Omega} \la \ast (d^A\lambda) +\nabla H(A), v\ra \,d{\rm vol} + \int_{\partial\Omega}\la \beta\wedge (A-A^{(b)}\ra +\int_{\Omega}\la \beta\wedge F^A\ra.
    \end{equation}
    This is continuous in $(A,\lambda, v, \beta)$. 
\end{proof}

\subsubsection{Coordinate expressions}

We now express the integrand in the first term on the right of (\ref{E:hatSpHA}) in coordinates. Let $E_1,\ldots, E_{\dim\gog}$ be an orthonormal basis of $\gog$, and we use the coordinate system $(x^1, x^2, x^3)$ on $\Omega$.  Then we can write $A$ and $\lambda$ as 
\begin{equation}
    \begin{split}
        A &= \sum_{r=1}^3\sum_{Z=1}^{\dim\gog}A_{Zr}\, E_Z \otimes dx^r\\
        \lambda &= \sum_{r=1}^3\sum_{Z=1}^{\dim\gog}\lambda_{Zr}\, E_Z \otimes dx^r.
    \end{split}
\end{equation}
    Then 
\begin{equation}
\begin{split}
  \ast (d^A\lambda) +\nabla H(A) &=\ast\left(\partial_q\lambda_p dx^q\wedge dx^p + [A_q, \lambda_p]dx^q\wedge dx^p \right) +\partial_{Di} H(A)E_D\otimes dx^i \\
  &=\epsilon_{qpi}\left( \partial_q\lambda_p + [A_q, \lambda_p]\right)\otimes dx^i +\partial_{Di} H(A) E_{D} \otimes dx^i \\
  &= \left(\epsilon_{qpi}  \partial_q\lambda_{Dp}E_{D} + \epsilon_{rpi}A_{Cr} \lambda_{Zp}[E_{C},E_{Z}]\right)\otimes dx^i +\partial_{Di} H(A) E_{D} \otimes dx^i \\
   &= \left[\epsilon_{qpi} \partial_q\lambda_{Dp} + A_{Cr} \lambda_{Zp}2{\tilde\epsilon}_{CZD} \epsilon_{rpi}   +\partial_{Zk} H(A)\right] E_{D} \otimes dx^k\\
   & \qquad \hbox{where ${\tilde\epsilon}_{CZD}=\frac{1}{2}\la [E_C, E_Z], E_D\ra_{\gog}$}\\
   &= \left[\epsilon_{qpi}\partial_q\lambda_{Dp} +2\lambda_{Zp}{\tilde\epsilon}_{CZD}\epsilon_{pir}A_{Cr}  + \partial_{Di} H(A)\right] E_{D}\otimes dx^i.
  \end{split}
\end{equation}
By the definition here of ${\tilde\epsilon}_{CZD}$ and by (\ref{E:TripEJKL}), 
we have
\begin{equation}
    {\tilde\epsilon}_{CZD}= \epsilon_{CZD}\qquad\hbox{for $\gog=su(2)$.}
\end{equation}

Now let $n$ be the ``outward''  unit vector field normal to $\partial\Omega$; then the orientation on $\partial\Omega$ is the one that makes an orthonormal frame $(u_1, u_2)$ on $\partial\Omega$ positively oriented if 
$$u_1\wedge u_2\wedge n = e_1\wedge e_2\wedge e_3,$$  
where $(e_1, e_2, e_3)$ is the standard basis of $\R^3$. If $\alpha$ is a $2$-form on the surface $\partial\Omega$, then, since $(u_1, u_2, n)$ is orthonormal,  
$$\alpha=\alpha(u_1, u_2)u^1\wedge u^2= (\alpha\wedge n_qdx^q)(u_1, u_2, n)u^1\wedge u^2= (\alpha\wedge n_qdx^q)(e_1, e_2, e_3)u^1\wedge u^2,$$
where $(u^1, u^2)$ is the frame of $T^*(\partial\Omega)$ dual to $(u_1, u_2)$. We will denote $u^1\wedge u^2$ by $da$, viewed effectively  as a surface measure.
The expression for ${\hat S}_H[A,\lambda]$ in  (\ref{E:dShatAdt2}), when expressed in terms of coordinates and components, works out to
\begin{equation}\label{E:dShatAdt4}
\begin{split}
&{\hat S}_H[A,\lambda]\\
&=-\int_{\partial\Omega} \lambda_{Zp}A^{(b)}_{Zr} [(dx^p\wedge dx^r\wedge n_qdx^q)(e_1, e_2, e_3)]\,da  + \int_{\Omega} A_{Zr} (\partial_q\lambda_{Yp})\la E_Y, E_Z\ra dx^q\wedge dx^p\wedge dx^r\\
&
\qquad +\int_{\Omega}\frac{1}{2} \lambda_{Zp}A_{Bq}A_{Cr}\la E_Z, [E_B, E_C]\ra dx^p\wedge dx^q\wedge dx^r +\int_{\Omega} H(A)\,d{\rm vol}\\
&= -\int_{\partial\Omega} \lambda_{Zp} A^{(b)}_{Zr} n_q\epsilon_{prq} \, da +\int_{\Omega} A_{Zr}(\partial_q\lambda_{Zp})\epsilon_{qpr} \,dx +\frac{1}{2} \int_{\Omega}\lambda_{Zp}A_{Bq}A_{Cr}2\epsilon_{BCZ}\epsilon_{pqr} \,dx +\int_{\Omega} H(A)\,dx\\
&= \int_{\Omega} \left[ A_{Zr}\epsilon_{rqp}\partial_q\lambda_{Zp} +\lambda_{Zp}\epsilon_{pqr}\epsilon_{ZBC}A_{Bq}A_{Cr} +H(A) \right]\, dx+\int_{\partial\Omega} -\lambda_{Zp}  \epsilon_{prq}  A^{(b)}_{Zr}  n_{q}\, da
\end{split}
\end{equation}
where we used (\ref{E:TripEJKL}) for the second equality, and $dx$ denotes the volume form $d{\rm vol}=dx^1\wedge dx^2\wedge dx^3$.

\subsubsection{The dual action $S_H(\lambda)$}

We assume now that the function $H$ is such that for each $\lambda\in \mca_k$ (or a suitably large subset of it) there is an $A^{(H)}(\lambda)\in\mca_k$ for which
\begin{equation}\label{E:defAHl}
    \ast\left(d^{ {A^{(H)}(\lambda)}}\lambda\right) +\nabla H((A^{(H)}(\lambda))=0,
\end{equation}
and that the {\em dual to primal map}
$$\lambda\mapsto A^{(H)}(\lambda)$$
is directionally $C^1$ (Definition \ref{D:direcder}) as a mapping between Fr\'echet spaces. Then we have the dual-space action
\begin{equation}\label{def: dual functional}
    S_H(\lambda)= {\hat S}\bigl[A^{(H)}(\lambda), \lambda].
\end{equation}
The existence of $A^{(H)}(\lambda)$ satisfying (\ref{E:defAHl}) needs to be established for each choice of the function $H$. In the following section, we will work through this for a specific choice for $H$ in which case the dual functional can be written out  using the coordinate expression (\ref{E:CSformcoords}) for the Chern-Simons form.

\begin{theorem}\label{T:SHcrit}
    With notation and assumptions as above,
    $S_H$ has a critical point at $\lambda$ if and only if 
    \begin{equation}\label{E:critptS}
        \begin{split}
            F^{A^{(H)}(\lambda)} &=0\\
            A^{(H)}(\lambda)|_{\partial\Omega} &= A^{(b)}
        \end{split}
    \end{equation}
\end{theorem}

Note that  there is no boundary constraint on the field $\lambda$, and that the boundary value for $A$ given in the second equation in (\ref{E:critptS}) appears automatically as part of the condition for $\lambda$ to be a critical point of $S_H$.

\begin{proof} The function ${\hat S}_H$ is directionally $C^1$   and so, given our assumption that $\lambda\mapsto A^{(H)}(\lambda)$ is also $C^1$, it follows that Hamilton's Fr\'echet space chain rule \cite[Theorem 3.3.4]{Ham1982} holds (see (\ref{E:HamCR})). Therefore, we can
compute the directional derivative of $S_H$ at $\lambda$ along $\beta$:
\begin{equation}\label{E:derivSHlb}
    \begin{split}
    \frac{d}{dt}\Big|_{t=0}S_H(\lambda+t\beta) &=
        \frac{d}{dt}\Big|_{t=0}S\left(A^{(H)}(\lambda+t\beta), \lambda+t\beta\right)\\
        &= \partial_1S\left(A^{(H)}(\lambda),\lambda\right)\Big| \left( A^{(H)}\right)'(\lambda)|\beta +\partial_2S\left(A^{(H)}(\lambda),\lambda\right)|\beta\\
        &=0 +\partial_2S\left(A^{(H)}(\lambda), \lambda\right)|\beta\\
        &=\int_{\partial\Omega}\left\langle \beta\wedge \bigl(A^{(H)}(\lambda)-A^{(b)}\bigr)\right\rangle +\int_{\Omega}\la \beta\wedge F^{A^{(H)}(\lambda)}\ra.
    \end{split}
\end{equation}
Thus, $\lambda$ is a critical point of $S_H$ if and only if the last line in (\ref{E:derivSHlb}) is $0$ for all $\beta$, which is equivalent to the equations (\ref{E:critptS}).
    \end{proof}

\subsection{Quadratic $H$}\label{sec:quad_H}

In this section we develop the explicit form of $S_H[\lambda]$ for the specific choice of a `shifted' quadratic form for $H$.
The language is that of simple vector calculus. We write field arguments of functionals within square brackets. 

As seen in (\ref{E:CSformcoords}), the    Chern-Simons action integral for a domain $\Omega \subset \R^3$ is given by:
\begin{equation*}
CS[A] = \int_\Omega 2 \eps_{pqr} \left( A_{Zp} \p_q A_{Zr} + \frac{2}{3} \eps_{ZBC} A_{Zp} A_{Bq} A_{Cr} \right) \, dx,
\end{equation*}
where we assume Dirichlet boundary conditions $A(x) = A^{(b)} (x), x \in \p \Omega$.

Recall from (\ref{E:ELCS1}) that the Euler Lagrange equations for the  Chern-Simons action integral are given by the flatness condition:
\begin{equation}\label{eq:C-S E-L}
    \eps_{pqr} \p_q A_{Zr} +  \eps_{pqr} \eps_{ZBC} A_{Bq} A_{Cr} = 0.
\end{equation}
In addition to $A$, we now introduce $\lambda$, also an $su(2)$-valued $1$-form; thus, we work with
\[
A, \lambda : \quad \R^{3} \supset \Omega \to \R^{3 \times 3}
\]
and define the Lagrangian $\scl$ of the pre-dual functional as the weak form of \eqref{eq:C-S E-L} with $\lambda$ as a test function, along with an additive shifted quadratic potential for $H$; as well, all boundary conditions of the primal problem are imposed as natural boundary condition terms of the functional: 
\begin{equation}\label{E:hatSAlambda}
\begin{aligned}
     \widehat{S}[A, \lambda] & = \int_{\Omega} A_{Zr}\,\eps_{rqp} \, \p_q \lambda_{Zp} +   \lambda_{Zp} \, \eps_{pqr} \, \eps_{ZBC} \, A_{Bq} \, A_{Cr} + \half k \, |A - \bar{A}|^2 \, dx + \int_{\p \Omega} -\lambda_{Zp} \epsilon_{pqr} A_{Zr}^{(b)} n_q \, da\\
     & =: \int_{\Omega} \scl( A, \lambda, \curl \,  \lambda) \, dx + \int_{\p \Omega}  - \lambda : (A^{(b)} \times n)\, da.
\end{aligned}
\end{equation}
Here, $n$ is the outward unit normal of $\partial \Omega$, the cross-product is row-wise, $da$ is the volume form on $\partial\Omega$, and the symbol $:$ represents contraction on two indices. (In more detail, $n_p$ is the unit normal to  $T_p\partial\Omega$ such that $p+tn_p\notin\Omega$ for all positive $t$, and $da$ is the Riemannian volume form on $\partial \Omega$.) 

Now,
\begin{equation*}
\frac{\p \scl}{\p A_{Di}} = 0 \Longleftrightarrow \quad k \, \left( \delta_{CD} \, \delta_{ir} + \frac{2}{k}  \lambda_{Zp} \, \eps_{pir} \, \eps_{ZDC}\right) (A_{Cr} - \bar{A}_{Cr}) = - \eps_{iqp} \, \p_q \lambda_{Dp} -2 \lambda_{Zp} \, \eps_{pir} \, \eps_{ZDC} \, \bar{A}_{Cr}.
\end{equation*}
So that the DtP mapping is defined through
\begin{equation}\label{E:kKDtP}
    k \, \K_{DiCr}\left(A^{(H)}_{Cr} - \bar{A}_{Cr} \right) = - \eps_{iqp} \, \p_q \lambda_{Dp} -2 \lambda_{Zp} \, \eps_{pir} \, \eps_{ZDC} \, \bar{A}_{Cr}; \qquad \K_{DiCr} : = \delta_{DC} \, \delta_{ir} + \frac{2}{k}  \lambda_{Zp} \, \eps_{pir} \, \eps_{ZDC}
\end{equation} 
The dual functional can be expressed as, \emph{with the notation $A^{(H)} = {\sf A}$},
\begin{equation}\label{E:preduShat}
\begin{aligned}
    S[\lambda] & = \widehat{S}[A^{(H)}, \lambda] \\
      & = \int_{\Omega} dx \ {\sf A}_{Zr} \, \eps_{rqp} \, \p_q \lambda_{Zp} +  \lambda_{Zp} \, \eps_{pqr} \, \eps_{ZBC} \, {\sf A}_{Bq} \, {\sf A}_{Cr} + \half k \, |{\sf A} - \bar{A}|^2  - \int_{\p \Omega} \lambda : (A^{(b)} \times n)\, d a \\
     & = \int_\Omega dx \   \eps_{rqp} \, \p_q \lambda_{Zp} \, \big({\sf A}_{Zr} - \bar{A}_{Zr} \big)\\
     & \qquad \qquad + \bar{A}_{Zr} \, \eps_{rqp} \, \p_q \lambda_{Zp}\\
     & \qquad \qquad +  \lambda_{Zp} \, \eps_{pqr} \, \eps_{ZBC} \big( {\sf A}_{Bq} - \bar{A}_{Bq}\big) \big( {\sf A}_{Cr} - \bar{A}_{Cr}\big)\\ 
     & \qquad \qquad +   \lambda_{Zp} \, \eps_{pqr} \, \eps_{ZBC} \big( {\sf A}_{Bq} - \bar{A}_{Bq}\big) \bar{A}_{Cr}\\
     & \qquad \qquad +   \lambda_{Zp} \, \eps_{pqr} \, \eps_{ZBC} \, \bar{A}_{Bq} \big( {\sf A}_{Cr} - \bar{A}_{Cr}\big) \\
     & \qquad \qquad -   \lambda_{Zp} \, \eps_{pqr} \, \eps_{ZBC} \, \bar{A}_{Bq} \, \bar{A}_{Cr}\\
     & \qquad \qquad +   \lambda_{Zp} \, \eps_{pqr} \, \eps_{ZBC} \, \bar{A}_{Bq} \, \bar{A}_{Cr} +  
     \lambda_{Zp} \, \eps_{pqr} \, \eps_{ZBC} \, \bar{A}_{Bq} \, \bar{A}_{Cr}\\
     & \qquad \qquad + \half k \, \big( {\sf A}_{Zr} - \bar{A}_{Zr} \big) \big( {\sf A}_{Zr} - \bar{A}_{Zr} \big)\\
     & \qquad \qquad -  \int_{\p \Omega} \lambda : (A^{(b)} \times n)\, d a.
\end{aligned}
\end{equation}
The linear terms in $({\sf A} - \bar{A})$ are
\begin{equation*}
    \begin{aligned}
    & \eps_{rqp} \, \p_q \lambda_{Zp} \, \big({\sf A}_{Zr} - \bar{A}_{Zr} \big) +   \lambda_{Zp} \, \eps_{pqr} \, \eps_{ZBC} \big( {\sf A}_{Bq} - \bar{A}_{Bq}\big) \bar{A}_{Cr} +  \lambda_{Zp} \, \eps_{pqr} \, \eps_{ZBC} \, \bar{A}_{Bq} \big( {\sf A}_{Cr} - \bar{A}_{Cr}\big) \\
    & = ({\sf A}_{Di} - \bar{A}_{Di}) \, \eps_{iqp} \, \p_q \lambda_{Dp} +   \lambda_{Zp} \, \eps_{pir} \, \eps_{ZDC} \, \bar{A}_{Cr} ({\sf A}_{Di} - \bar{A}_{Di}) +   \lambda_{Zp} \, \eps_{pri} \, \eps_{ZCD} \, \bar{A}_{Cr} \big( {\sf A}_{Di} - \bar{A}_{Di}\big) \\
    & = ({\sf A}_{Di} - \bar{A}_{Di}) \, \eps_{iqp} \, \p_q \lambda_{Dp} + 2\lambda_{Zp} \, \eps_{pir} \, \eps_{ZDC} \, \bar{A}_{Cr} ({\sf A}_{Di} - \bar{A}_{Di})\\
    & = -  ( - \eps_{iqp} \, \p_q \lambda_{Dp} -2 \lambda_{Zp} \, \eps_{pir} \, \eps_{ZDC} \, \bar{A}_{Cr} ) ({\sf A}_{Di} - \bar{A}_{Di})
    \end{aligned}
\end{equation*}

The quadratic terms in $({\sf A} - \bar{A})$ are
\begin{equation*}
    \begin{aligned}
        &     \lambda_{Zp} \, \eps_{pqr} \, \eps_{ZBC} \big( {\sf A}_{Bq} - \bar{A}_{Bq}\big) \big( {\sf A}_{Cr} - \bar{A}_{Cr}\big) + \half k \, \big( {\sf A}_{Zr} - \bar{A}_{Zr} \big) \big( {\sf A}_{Zr} - \bar{A}_{Zr} \big) \\
        & =   \lambda_{Zp} \, \eps_{pir} \, \eps_{ZDC} \big( {\sf A}_{Di} - \bar{A}_{Di}\big) \big( {\sf A}_{Cr} - \bar{A}_{Cr}\big) + \half k \, \delta_{CD} \, \delta_{ir} \, \big( {\sf A}_{Di} - \bar{A}_{Di} \big) \big( {\sf A}_{Cr} - \bar{A}_{Cr} \big)\\
        & = \half k \left( \, \delta_{CD} \, \delta_{ir} \, + \frac{2}{k} \lambda_{Zp} \, \eps_{pir} \, \eps_{ZDC} \right) \big( {\sf A}_{Di} - \bar{A}_{Di} \big) \big( {\sf A}_{Cr} - \bar{A}_{Cr} \big)\\
        & = \half k \, \K_{DiCr} \big( {\sf A}_{Di} - \bar{A}_{Di} \big) \big( {\sf A}_{Cr} - \bar{A}_{Cr} \big)\\
        & = \half ( - \eps_{iqp} \, \p_q \lambda_{Dp} -2 \lambda_{Zp} \, \eps_{pir} \, \eps_{ZDC} \, \bar{A}_{Cr} ) \big( {\sf A}_{Di} - \bar{A}_{Di} \big);
    \end{aligned}
\end{equation*}
where we used (\ref{E:kKDtP}) in the last step.
So, the linear plus quadratic terms in $({\sf A} - \bar{A})$ are given by
\begin{equation*}
    \begin{aligned}
         & \left( - 1 + \half \right) \big( - \eps_{iqp} \, \p_q \lambda_{Dp} -2 \lambda_{Zp} \, \eps_{pir} \, \eps_{ZDC} \, \bar{A}_{Cr} \big) \big({\sf A}_{Di} - \bar{A}_{Di} \big)\\
        & = - \half \big( - \eps_{iqp} \, \p_q \lambda_{Dp} - 2\lambda_{Zp} \, \eps_{pir} \, \eps_{ZDC} \, \bar{A}_{Cr} \big) \big({\sf A}_{Di} - \bar{A}_{Di} \big),
    \end{aligned}
\end{equation*}
and since
\begin{equation}\label{E:AminusbarA}
   \big({\sf A}_{Di} - \bar{A}_{Di} \big) = \frac{1}{k} \left( \K^{-1} \right)_{DiCr} \big( - \eps_{rgh} \, \p_g \lambda_{Ch} - 2\lambda_{Bj} \, \eps_{jrn} \, \eps_{BCI} \, \bar{A}_{In}   \big) 
\end{equation} 
from the DtP mapping, the linear plus quadratic terms become

\begin{equation*}
   \begin{aligned}
- \half  & \big( - \eps_{iqp} \, \p_q \lambda_{Dp} -2 \lambda_{Zp} \, \eps_{pir} \, \eps_{ZDC} \, \bar{A}_{Cr} \big) \  \frac{1}{k} \left( \K^{-1} \right)_{DiCr} \big( - \eps_{rgh} \, \p_g \lambda_{Ch} - 2\lambda_{Bj} \, \eps_{jrn} \, \eps_{BCI} \, \bar{A}_{In}).
\end{aligned} 
\end{equation*}

The `constant' terms in $({\sf A} - \bar{A})$ are given by
\[
\left(- 1+1+1 \right) \lambda_{Zp} \, \eps_{pqr} \, \eps_{ZBC} \, \bar{A}_{Bq} \, \bar{A}_{Cr}  =   \lambda_{Zp} \, \eps_{pqr} \, \eps_{ZBC} \, \bar{A}_{Bq} \, \bar{A}_{Cr}.
\]
Thus, the dual Chern-Simons functional for $H = \half k |A - \bar{A}|^2$ is given by

\begin{equation}\label{eq:CS_quad_expl_dual}
    \begin{aligned}
    {\sf P}_{Di} & := - \eps_{iqp} \, \p_q \lambda_{Dp} - 2\lambda_{Zp} \, \eps_{pir} \, \eps_{ZDC} \, \bar{A}_{Cr}\\
    \K_{DiCr} & := \delta_{DC} \, \delta_{ir} + \frac{2}{k}  \lambda_{Zp} \, \eps_{pir} \, \eps_{ZDC}\\
        S[\lambda] & = \bigintsss_\Omega \half \left( - {\sf P}_{Di} \ \frac{1}{k} \left( \K^{-1} \right)_{DiCr} \  {\sf P}_{Cr} + 2\lambda_{Zp} \, \eps_{pqr} \, \eps_{ZBC} \, \bar{A}_{Bq} \, \bar{A}_{Cr} \right) \, dx - \int_{\p \Omega} \lambda : (A^{(b)} \times n)\, d a.\\
    \end{aligned}
\end{equation}
We note that if $\bar{A}$ is a solution to \eqref{eq:C-S E-L}, then $\lambda = 0$ is a critical point of $S$.

\section{Variational Analysis}\label{sec:var_anal}

In this section, we present a variational analysis for the E-L equation of the C-S functional based on the Direct Method of the Calculus of Variations for an appropriately defined dual functional related to the dual action functional $S_H$, \eqref{def: dual functional}. 

Let us first note that the C-S-functional itself is not suitable for applying the Direct Method of the Calculus of Variations in order to establish existence of critical points by proving the existence of a minimizer.   
Indeed, let $\varphi \in C^{\infty}_c(\Omega;[0,\infty))$ with $\int_{\Omega} \varphi(x)^3 \, dx > 0$ and $A^t_{Jk}(x) = t \varphi(x) \delta_{Jk}$, for $t \in \R$, then it  follows by \eqref{E:CSformcoords} that:
\begin{align*}
CS(A^t) &= 2 \int_{\Omega} \eps_{pqr} (A^t_{Jp} \partial_q A^t_{Jr} + \frac23 \eps_{JKL} A^t_{Jp} A^t_{Kq} A^t_{Lr}) \, dx \\
&= 2\int_{\Omega} \frac23 t^3 \varphi(x)^3 \eps_{pqr} \eps_{JKL} \delta_{Jp} \delta_{Kq} \delta_{Lr} \, dx \\
&= 8 t^3 \underbrace{\int_{\Omega} \varphi(x)^3 \, dx}_{>0}.
\end{align*}
By sending $t \to \pm\infty$ it is evident that the $CS$ functional cannot have a minimizer or a maximizer with finite energy in any function space that includes $C^{\infty}_c(\Omega;\R^{3\times 3})$. 

In this section, we consider the dual functional
\begin{equation}\label{eq: def tildeS}
\tilde{S}_H [\lambda] = \sup_{A \in X} \int_{\Omega} A: \operatorname{curl} \lambda  +  \lambda_{Zp} \eps_{pqr} \eps_{ZDC} A_{Dq} A_{Cr} - H(A) \, dx -   \int_{\partial \Omega} \lambda : ( A^{(b)} \times n) \, da, 
\end{equation}
where $n$ is the outer normal to $\partial \Omega$ and the function space $X$ and the domain of $\tilde{S}_H$ will be chosen such that the integral is well-defined. For regular functions $\lambda$ and $A^{(b)}$ the integral over the boundary $\partial \Omega$ is with respect to the surface measure on $\partial \Omega$. For less regular functions this term has to be understood in the sense of traces as will be explained once the function spaces for $A$ and $\lambda$ will be chosen for specific choices of $H$.  

Let us briefly discuss the connection between \eqref{eq: def tildeS} and the approach introduced in Section \ref{sec:geom_back}.
Consider
\begin{equation*}
    \begin{aligned}
        T_{Zp}(A) &: = \eps_{pqr} \eps_{ZDC} A_{Dq} A_{Cr} \\
        B(A, \lambda, n) & : = - \lambda : (A \times n) = - \lambda_{Zp} \eps_{prq} A_{Zr} n_q = \lambda_{Zp} \eps_{pqr} A_{Zr} n_q = A_{Zr} \eps_{rpq} \lambda_{Zp}  n_q =: A: (\lambda \times n)
    \end{aligned}
\end{equation*}
Then, if $A \mapsto H(A) \in \R$ is a strictly convex function with positive definite Hessian everywhere,
\begin{equation*}
    \begin{aligned}
      & \bullet \quad \sup_{\lambda}\left( \inf_{A} \int_\Omega A: \curl \lambda + \lambda : T(A) + H(A) \, dx  \quad + \quad \int_{\p \Omega} B(A^{(b)}, \lambda, n) \, da \right) \\
        &  \bullet \quad \inf_{\lambda }\left( \sup_{A} \int_\Omega - A: \curl \lambda - \lambda : T(A) - H(A) \, dx  \quad - \quad \int_{\p \Omega} B(A^{(b)}, \lambda, n) \, da \right) \\
       &    \bullet \quad \inf_{\lambda } \left( \sup_{A } \int_\Omega A: \curl \lambda + \lambda : T(A) - H(A) \, dx  \quad + \quad \int_{\p \Omega} B(A^{(b)}, \lambda, n) \, da \right) \\  
       &  \bullet \quad \sup_{\lambda} \left( \inf_{A} \int_\Omega - A: \curl \lambda - \lambda : T(A) + H(A) \, dx  \quad - \quad \int_{\p \Omega} B(A^{(b)}, \lambda, n) \, da \right)
    \end{aligned}
\end{equation*}
are four\footnote{We note that the second statement is the negative of the first and the fourth is the negative of the third. The second statement is equal to the third  whenever for each $\lambda$ in the set to which it belongs, $-\lambda$ also belongs to the set.} viable formulations of our dual scheme whose corresponding minimizers (or maximizers) can be used to define variational dual solutions (see Definition \ref{def: dual sol} below) for the Euler-Lagrange equations of the Chern-Simons action. In this Section, we analyze the third option above. The functional $S[\lambda]$ defined in Sec.~\ref{sec:quad_H} is related to the first option above with its maximizer defining weak solutions to the C-S Euler-Lagrange system in the presence of appropriate regularity. The viability arises from the fact that the inner $\sup_A$ ($\inf_A$) may be thought of as formally equivalent to evaluating $\scl(A(x), \dee(x))$ at its critical point w.r.t $A(x)$ for each fixed $\dee(x)$ (the latter generated from an  appropriate class of $\lambda$ fields); such a critical point is defined by $A^{(H)}(\dee(x))$. 
Hence, if the dual functional $\tilde{S}$ is regular enough in a minimizer, $x \mapsto A^{(H)}(\dee(x))$ will satisfy the Euler-Lagrange equation of $\tilde{S}$ which, by the discussion in Section \ref{ss:gdds}, is exactly the Euler-Lagrange equation of the Chern-Simons action functional \eqref{eq:C-S E-L}.

In light of the discussion above, we define the notion of a \textit{variational dual solution} and a \textit{dual solution}, c.f.~\cite{AcharyaGinsterSingh}.
\begin{definition} \label{def: dual sol} Given a Borel-measurable function $H$ we say that $\lambda$ is a variational dual solution to the Chern-Simons equation if it minimizes the dual functional $\tilde{S}_H$. 
Moreover, we say that $\lambda$ is a dual solution to the Chern-Simons equation if the dual functional $\tilde{S}_H$ is differentiable at $\lambda$, $\lambda$ satisfies the weak form of the corresponding Euler-Lagrange equation for $\tilde{S}_H$ and there exists a smooth dual-to-primal mapping $\lambda \to A^{(H)}(\lambda)$ corresponding to $\tilde{S}_H$.
\end{definition}
First, by the discussion above every dual solution gives rise to a weak solution of the Chern-Simons equation.
The notion of variational dual solutions is a weaker concept and not every variational dual solution is necessarily a dual solution or gives rise to a weak solution of  the Chern-Simons equation. However, the choice of the function $H$ does significantly influence this relationship, see also the discussion in \cite[Section 5]{AcharyaGinsterSingh}.

In this work, we focus on the existence of variational dual solutions for reasonable choices of $H$, i.e., we show the following:
\begin{theorem}\label{thm: existence}
    For $H(A) = \ell |A|^{\alpha}$, $\ell>0$, $\alpha \geq 2$, there exists a variational dual solution to the Chern-Simons equation.
\end{theorem}
\begin{remark}
    It can be seen from the proof that the statement of the theorem stays true for all convex functions $H$ such that $|H(A)| \leq \ell \left( |A|^{\alpha} + 1 \right)$, c.f.~also Remark \ref{rem: mixed}. Note that this includes the case $1 \leq  \alpha < 2$. However, in this case the dual functional is degenerate in the sense that $\tilde{S}_H[\lambda] = +\infty$ for all $\lambda \neq 0$. Hence, a variational dual solution exists but is not meaningful.
\end{remark}
In particular, this shows that the problem of finding an \emph{extremal} point for the Chern-Simons functional can (at least in a weak sense) be replaced by finding a \emph{minimizer} of an appropriate dual functional.

By the definition of $\tilde{S}_H$ it can be checked after fixing appropriate function spaces that $\tilde{S}_H$ is convex and lower semi-continuous with respect to the appropriate weak (or weak*) convergence of $\lambda$ and $\operatorname{curl} \lambda$ in these spaces. This will be made more precise below once the function spaces and $H$ will be fixed. Hence, in order to prove Theorem \ref{thm: existence}, i.e., to establish the existence of varitional dual solutions, using the Direct Method of the Calculus of Variations the key step is to show that $\tilde{S}_H$ is (weakly) coercive on an appropriate function space for $\lambda$. This will be done in Proposition \ref{prop: coerc higher} (for $\alpha > 2)$ and Proposition \ref{prop: coerc quadratic} (for $\alpha=2$). For a more detailed discussion we refer to \cite[Section 5]{AcharyaGinsterSingh}.

In order to simplify notation we define for a convex function $H: \R^{3\times 3} \to \R$, the function $g_H: \R^{3 \times 3} \times \R^{3\times 3} \to \R \cup \{+\infty\}$ as
\begin{equation*}
g_H(\lambda,\mu) = \sup_{A \in \R^{3\times 3}} A:\mu +  \lambda_{Zp} \eps_{pqr} \eps_{ZBC} A_{Bq} A_{Cr} - H(A).
\end{equation*}
As in \cite{AcharyaGinsterSingh} it can then be argued that it holds in the setting discussed below that 
\begin{equation*}
\tilde{S}_H[\lambda] = \int_{\Omega} g_H(\lambda,\operatorname{curl} \lambda) \, dx - \int_{\partial \Omega} \lambda : ( A^{(b)} \times n) \, da.
\end{equation*}

\subsection{The case $\alpha > 2$}

Let $\infty > \alpha > 2$, $\alpha' := \frac{\alpha}{\alpha-1}$ and $\beta := \left(\frac{\alpha}2\right)' = \frac{\alpha}{\alpha-2}$. Note that $\beta \geq \alpha'$ since $\alpha/2 \leq \alpha$. 

The case of $\alpha = 4$ is worked out in Appendix \ref{app:quartic} with slightly modified constants as a gentle example of the general case dealt with herein.

In the case of $\alpha$-growth of $H$, the natural space for the primal variables $A$ is $X = L^{\alpha}(\Omega;\R^{3\times 3})$ and, accordingly, $(\lambda,\curl \lambda) \in L^{\beta}(\Omega; \R^{3\times 3}) \times L^{\alpha'}(\Omega;\R^{3\times3})$, where $\operatorname{curl} \lambda$ has to be understood in the sense of distributions. Note that these choices guarantee by H\"older's inequality that the integral 
$$\int_{\Omega} A:\operatorname{curl} \lambda  +  \lambda_{Zp} \eps_{pqr} \eps_{ZBC} A_{Bq} A_{Cr} - H(A) \, dx$$
is well-defined. 

On the other hand, let $\lambda \in L^{\beta}$ with $\operatorname{curl} \lambda \in L^{\alpha'}$ and a measurable function $A: \Omega \to \R^{3\times 3}$ such that $g_H(\lambda(x), \operatorname{curl} \lambda(x)) = A(x):\operatorname{curl} \lambda(x)  +  \lambda_{Zp}(x) \eps_{pqr} \eps_{ZBC} A_{Bq}(x) A_{Cr}(x) - H(A(x))$. Then the first order optimality condition for $A$ reads as
\[
0 = (\operatorname{curl }\lambda(x))_{ij} + \lambda_{Zp}(x) \eps_{pjr} \eps_{ZiC} A_{Cr}(x) + \lambda_{Zp}(x) \eps_{pqj} \eps_{ZBi} A_{Bq}(x)  - (D H (A(x)))_{ij}.
\]
Since $|D H (A(x))| = \alpha \ell |A(x)|^{\alpha - 1}$ it follows that
\[
|A(x)|^{\alpha - 1} \leq C'\left( |\operatorname{curl} \lambda(x)| + |\lambda(x)| \cdot |A(x)| \right),
\]
which implies that 
\[
|A(x)|^{\alpha} \leq \tilde{C} \left( |\operatorname{curl} \lambda(x)|^{\frac{\alpha}{\alpha -1 }} + |\lambda(x)|^{\frac{\alpha}{\alpha-2}} \right) = \tilde{C} \left( |\operatorname{curl} \lambda(x)|^{\alpha'} + |\lambda(x)|^{\beta} \right) \in L^1(\Omega),
\]
i.e., $A \in L^{\alpha}(\Omega;\R^{3\times 3})$.

Next, note that for functions  $\lambda \in L^{\beta}(\Omega; \R^{3\times 3}) \subseteq L^{\alpha'}(\Omega; \R^{3\times 3})$ with $\operatorname{curl} \lambda \in L^{\alpha'}(\Omega; \R^{3\times 3})$ there exists a linear, continuous trace operator for the term $\lambda \times n$ into the space $W^{-1+ \frac1{\alpha},\alpha'}(\partial \Omega; \R^{3\times 3})$, the dual of the space $W^{1-\frac1{\alpha},\alpha}(\partial \Omega; \R^{3\times 3})$ which in turn is the trace space of $W^{1,\alpha}(\Omega,\R^{3\times 3})$, c.f.~see \cite[Theorem 18.40]{LeoniSobolev} and \cite[Theorem 1, page 204]{Lions} for the case $\alpha =2$ (the other cases are similar). In particular, assuming that $A^{(b)} \in W^{1-\frac1{\alpha},\alpha}(\partial \Omega; \R^{3\times 3})$ (i.e., $A^{(b)}$ is the trace of a function in $W^{1,\alpha}(\Omega;\R^{3\times 3})$) then the well-defined interpretation of the integral $\int_{\partial \Omega} A^{(b)} : (\lambda \times n) \, da$ is the dual pairing $\langle A^{(b)}, \lambda \times n \rangle_{W^{1-\frac1{\alpha},\alpha}(\partial \Omega; \R^{3\times 3}),W^{-1+\frac1{\alpha},\alpha'}(\partial \Omega; \R^{3\times 3})}$. This expression is continuous with respect to the weak convergences $\lambda_k \rightharpoonup \lambda$ in $L^{\beta}$ and $\operatorname{curl} \lambda_k \rightharpoonup \operatorname{curl} \lambda$ in $L^{\alpha'}$. Moreover, it holds 
\begin{align} \nonumber 
    \left| \langle A^{(b)}, \lambda \times n \rangle_{W^{1-\frac1{\alpha},\alpha},W^{-1+\frac1{\alpha},\alpha'}} \right| &\leq \|A^{(b)} \|_{W^{1-\frac1{\alpha},\alpha}} \| \lambda \times n\|_{W^{-1+\frac1{\alpha},\alpha'}} \\ &\leq C' \|A^{(b)} \|_{W^{1-\frac1{\alpha},\alpha}} \left( \| \lambda\|_{L^{\beta}} + \| \operatorname{curl } \lambda \|_{L^{\alpha'}} \right). \label{eq:higher_trace}
\end{align} 

It can then be shown that in this setting $\tilde{S}_H$ is lower-semicontinuous with respect to $\lambda_k \rightharpoonup \lambda$ in $L^\beta$ and $\operatorname{curl} \lambda_k \to \operatorname{curl} \lambda$ in $L^{\alpha'}$.
Hence, the main step to prove existence of a minimizer for $\tilde{S}_H$ using the Direct Method of the Calculus of Variations is to prove coercivity of $\tilde{S}_H$.

Note that in order to show coercivity of $\tilde{S}_H$ for all $H(A) = \ell |A|^{\alpha}$, $\ell >0$, it suffices to consider one specific choice for $\ell$. 
Indeed, for general $\ell'$ one uses for $\tilde{\ell} = \frac{\ell'}{\ell} $ that 
\begin{align}
&\sup_{A}   A:\mu +  \lambda_{Zp} \eps_{pqr} \eps_{ZBC} A_{Bq} A_{Cr} - \ell' |A|^{\alpha} \\ = &\sup_A \tilde{\ell} \left( A: (\tilde{\ell}^{-1}\mu) +  (\tilde{\ell}^{-1}\lambda_{Zp}) \eps_{pqr} \eps_{ZBC} A_{Bq} A_{Cr} - \ell |A|^{\alpha}  \right).
\end{align}
Assuming that coercivity holds for $\tilde{S}_H$ for the specific $\ell>0$, the above yields also coercivity with respect to $\tilde{\ell}^{-1} \lambda$ and $\tilde{\ell}^{-1} \operatorname{curl} \lambda$ for general $\ell'$ which in turn implies coercivity with respect to $\lambda$ and $\operatorname{curl} \lambda$. For technical reasons that become apparent in the proof, we will consider $\ell= \frac1{\alpha \cdot 2^{\alpha/2} \cdot 16 \cdot 3^{\beta}} \left(4 \alpha'\right)^{-\frac{\alpha}{2-\alpha'}}$.

The main step to show the coercivity of $\tilde{S}_H$ is to prove prove lower bounds on the function $g_H$.

\begin{proposition}\label{prop: lb higher}
    Let $H(A) = \frac1{\alpha \cdot 2^{\alpha/2} \cdot 16 \cdot 3^{\beta}} \left(4 \alpha'\right)^{-\frac{\alpha}{2-\alpha'}} |A|^{\alpha}$. Then there exists $c>0$ such that 
    \[
    g_H(\lambda, \mu) \geq c \left( |\mu|^{\alpha'} + |\lambda|^{\beta} \right).
    \]
\end{proposition}

\begin{proof}
    Fix $\lambda, \mu \in \R^{3\times 3}$. 
    We distinguish the following cases.
    \begin{enumerate}
        \item Assume that $|\lambda| \leq \frac1{4\alpha'} |\mu|^{ 2 - \alpha'}$. 
        Let $A = \mu |\mu|^{\alpha'-2}$ and write 
        $$K_{Zp} = \eps_{pqr} \eps_{ZBC} A_{Bq} A_{Cr} =  \eps_{ZDC} (A_D \times A_C)_p.$$
        First, note by Young's inequality that
        \begin{align}
            |K|^2 &= 4 |A_1 \times A_2|^2 + 4 |A_2 \times A_3|^2 + 4 |A_3 \times A_1|^2 \nonumber \\
            &\leq 4 |A_1|^2 |A_2|^2 + 4 |A_2|^2 |A_3|^2 +  4 |A_3|^2 |A_1|^2 \leq 4| A_1|^4 +  4 |A_2|^4 + 4 |A_3|^4 \leq 4 |A|^4. \label{eq: est K}
        \end{align}
        Since, $\alpha (\alpha' - 1) = \frac{\alpha}{\alpha - 1} = \alpha'$ it follows
        \begin{align*}
            g_H(\lambda, \mu) \geq \mu:A - 2 |\lambda| |A|^2 - \frac1{\alpha}|A|^{\alpha} &= \frac1{\alpha'} |\mu|^{\alpha'} - 2 |\lambda| |\mu|^{2 (\alpha'-1)} \\ & \geq \frac1{2\alpha'} |\mu|^{\alpha'} \geq \frac1{4 \alpha'} \left( |\mu|^{\alpha'} + |\lambda|^{\beta} \right). 
        \end{align*}
        For the last estimate we used that $|\lambda|^{\beta} \leq \frac1{(4\alpha')^{\beta}} |\mu|^{(2 - \alpha')\beta} \leq |\mu|^{\alpha'}$ as $(2-\alpha')\beta = \frac{\alpha - 2}{\alpha- 1} \beta = \frac{\alpha}{\alpha-1} = \alpha'$.
        \item Assume that $\frac1{4\cdot 3^{\beta/2}}  |\lambda|^{\frac1{2-\alpha'}} \leq |\mu| \leq \left( 4\alpha' \right)^{\frac1{2-\alpha'}} |\lambda|^{\frac1{2-\alpha'}}$. 
        For $A = (4|\lambda|)^{-1} \mu$ we find using $\frac1{2-\alpha'} = \frac{\alpha-1}{\alpha-2} = \frac{\beta}{\alpha'}$, $\alpha \left(\frac{1}{2-\alpha'} - 1\right)= \frac{\alpha}{\alpha-2} = \beta$ and $\frac{2}{2-\alpha'} - 1 = \frac{\alpha'}{2-\alpha'} = \beta$ 
        \begin{align*}
            g_H(\lambda, \mu) &\geq \mu: A - 2 |\lambda| |A|^2 - \frac1{16 \cdot 3^{\beta}} \left(4 \alpha'\right)^{-\frac{\alpha}{2-\alpha'}} |A|^{\alpha} \\
            &= \frac14 |\lambda|^{-1} |\mu|^2 - \frac1{8} |\lambda|^{-1} |\mu|^2 - \frac1{4^{\alpha} \cdot 16 \cdot 3^{\beta}} \left(4 \alpha'\right)^{-\frac{\alpha}{2-\alpha'}} |\lambda|^{-\alpha} |\mu|^{\alpha} \\
            &\geq \frac1{8\cdot 16 \cdot 3^{\beta}} |\lambda|^{-1+ \frac{2}{2-\alpha'}} - \frac1{16 \cdot 16 \cdot 3^{\beta}} \left(4 \alpha'\right)^{-\frac{\alpha}{2-\alpha'}} |\lambda|^{-\alpha} |\mu|^{\alpha} \\
            &\geq \frac1{8 \cdot 16 \cdot 3^{\beta}} |\lambda|^{\beta} -  \frac1{16 \cdot 16 \cdot 3^{\beta}} \left(4 \alpha'\right)^{-\frac{\alpha}{2-\alpha'}} \left(4 \alpha'\right)^{\frac{\alpha}{2-\alpha'}} |\lambda|^{\alpha (-1 + \frac1{2-\alpha'})}   \\
            &= \frac1{8 \cdot 16 \cdot 3^{\beta}} |\lambda|^{\beta} -  \frac1{16 \cdot 16 \cdot 3^{\beta}} |\lambda|^{\beta} \\
            &= \frac1{16 \cdot 16 \cdot 3^{\beta}}|\lambda|^{\beta} \geq \frac1{32 \cdot 16 \cdot 3^{\beta}}(4 \alpha')^{-\frac{\alpha'}{2-\alpha'}}  \left( |\lambda|^{\beta} + |\mu|^{\alpha'} \right).
        \end{align*}
In the last estimate we used that $|\mu|^{\alpha'} \leq (4 \alpha')^{\frac{\alpha'}{2-\alpha'}} |\lambda|^{\beta}$.
  
         \item Assume that $|\mu| \leq \frac1{4\cdot 3^{\beta/2}} |\lambda|^{\frac1{2-\alpha'}}$. Let $V \in \{1,2,3\}$ such that it holds for the $V^{th}$ row, $\Lambda_V \in \R^3$, of $\lambda \in \R^{3\times 3}$ that $ m^2:= \lambda_{Vi} \lambda_{Vi} \mbox{ (no sum on $V$) } = |\Lambda_V|^2  \geq \frac13 |\lambda|^2$, $m > 0$. 
        Choose an orthonormal basis $(E_1, E_2, E_3)$ with $E_V = \frac{1}{m}\Lambda_V$ and $\half \eps_{VYZ} E_Y \times E_Z = E_V$.
        Now, define $A \in \R^{3\times 3}$ row-wise as 
        \[
        A_{Yi} =  m^{(\beta-1)/2} \, E_{Yi} \quad \text{ if } Y \neq V \qquad \text{ and } \qquad A_{Vi} = 0.
        \]
It follows that $|A| = \sqrt{2}m^{(\beta-1)/2} = \sqrt{2}|\Lambda_V|^{(\beta-1)/2}$ and 
        \begin{equation}\label{E:epqrZp}
        \eps_{pqr} \eps_{ZBC} A_{Bq} A_{Cr}  = 
        \begin{cases} 2 m^{\beta - 1} \, E_{Vp} = 2 \lambda_{Vp} |\Lambda_V|^{\beta - 2}  &\text{ if } Z=V, \\ 0 &\text{ otherwise.} 
        \end{cases}
        \end{equation}
        Then we estimate using $\alpha \frac{\beta-1}2 = \beta$ and $\frac{\beta-1}2 + \frac{1}{2-\alpha'} =\frac1{\alpha-2} + \frac{\alpha-1}{\alpha-2} = \frac{\alpha}{\alpha-2} = \beta$ 
        \begin{align*}
            g_H(\lambda, \mu) &\geq -|\mu| |A| + 2 |\Lambda_V|^{\beta} - \frac1{2^{\alpha/2}} |A|^{\alpha} \geq - 2^{1/2} |\mu| |\Lambda_V|^{(\beta-1)/2} + |\Lambda_V|^{\beta} \\
            &\geq -2 |\mu| | \lambda|^{(\beta-1)/2} + \frac1{3^{\beta/2}} |\lambda|^{\beta} 
            \geq \frac1{2 \cdot 3^{\beta/2}} |\lambda|^{\beta}  \geq \frac1{4 \cdot 3^{\beta/2}} \left( |\lambda|^{\beta} + |\mu|^{\alpha'} \right).
        \end{align*}
    \end{enumerate}
\end{proof}

\begin{proposition}\label{prop: coerc higher}
    For $H(A) = \frac1{\alpha \cdot 2^{\alpha/2} \cdot 16 \cdot 3^{\beta}} \left(4 \alpha'\right)^{-\frac{\alpha}{2-\alpha'}} |A|^{\alpha} $ and $A^{(b)} \in W^{1-\frac1{\alpha},\alpha}(\partial \Omega;\R^{3\times3})$ the dual functional $\tilde{S}_H$ is weakly coercive on the Banach space $\mathcal{A}:=\{ \lambda \in L^{\beta}(\Omega; \R^{3\times 3}): \operatorname{curl}(\lambda) \in L^{\alpha'}(\Omega;\R^{3\times 3}) \}$, i.e., for every sequence $(\lambda_k)_k \subseteq \mathcal{A}$ such that $\sup_k \tilde{S}_H[\lambda_k] < \infty$ there exists a (not relabeled) subsequence and $\lambda \in \mathcal{A}$ such that $\lambda_k \rightharpoonup \lambda$ in $L^{\beta}$ and $\operatorname{curl} (\lambda_k)\rightharpoonup \operatorname{curl}(\lambda)$ in $L^{\alpha'}$.
\end{proposition}
\begin{proof}
    By Proposition \ref{prop: lb higher} and the trace estimate \eqref{eq:higher_trace} it follows for a sequence $(\lambda_k)_k \subseteq \mathcal{A}$ satisfying $\sup_k \tilde{S}_H[\lambda_k] \leq C < \infty$ for all $k\in \N$ that 
    \[
     c ( \| \lambda_k \|_{L^{\beta}}^{\beta} + \| \operatorname{curl } \lambda_k \|_{L^{\alpha'}}^{\alpha'} ) - C'  \|A^{(b)} \|_{W^{1-\frac1{\alpha},\alpha}} \left( \| \lambda_k \|_{L^{\beta}} + \| \operatorname{curl } \lambda_k \|_{L^{\alpha'}} \right) \leq \tilde{S}_H[\lambda_k] \leq  C.
    \]
    Since $\lim_{\rho, \sigma \to +\infty} c ( \rho^{\beta} + \sigma^{\alpha'} ) - C'  \|A^{(b)} \|_{W^{1-\frac1{\alpha},\alpha}} \left( \rho + \sigma \right) = + \infty$,
 it follows that the terms $\| \lambda_k\|_{L^{\beta}}$ and $ \| \operatorname{curl}(\lambda_k)\|_{L^{\alpha'}}$ are uniformly bounded. By the usual compactness arguments this implies the existence of a (not relabeled) subsequence and functions $\lambda \in L^{\beta}(\Omega;\R^{3\times 3})$ and $\mu \in L^{\alpha'}(\Omega;\R^{3\times 3})$ such that $\lambda_k \rightharpoonup \lambda$ in $L^{\beta}$ and $\operatorname{curl} (\lambda_k)\rightharpoonup \mu$ in $L^{\alpha'}$. Eventually, it can then be checked that actually $\operatorname{curl}(\lambda) = \mu$. 
\end{proof}

Note that by similar arguments one can show directly from Proposition \ref{prop: lb higher} and \eqref{eq:higher_trace} that the dual functional $\tilde{S}_H$ is bounded from below.

\begin{remark}\label{rem: mixed}
    Note that for $\tilde{H}(A) = \ell |A|^{\alpha} + \tilde{\ell} |A|^2$, $\ell, \tilde{\ell} >0$, $\alpha > 2$, the functional $\tilde{S}_{\tilde{H}}$ is also coercive on the space $\{ \lambda \in L^{\beta}(\Omega; \R^{3\times 3}): \operatorname{curl}(\lambda) \in L^{\alpha'}(\Omega;\R^{3\times 3}) \}$. Indeed, since $\ell |A|^{\alpha} + \tilde{\ell} |A|^2 \leq (\ell + \tilde{\ell}) |A|^{\alpha} + \tilde{\ell} =: H(A) + \tilde{\ell}$ it follows that $\tilde{S}_{\tilde{H}} \geq \tilde{S}_{H} - \tilde{\ell} |\Omega|$. Hence, the coercivity of $\tilde{S}_{\tilde{H}}$ follows from the shown coercvity of $\tilde{S}_{H}$. The advantage of using $\tilde{H}$ might be that its Hessian is strictly positive definite which facilitates that the DtP-mapping can at least locally be well-defined in a smooth way.
\end{remark}

\subsection{Quadratic $H$}

In this section we consider the quadratic function $H(A) = \ell |A|^2.$ With the same argument as in the case $\alpha > 2$ it suffices to consider $\ell=\frac12$.

In this setting, the natural space for the primal functions $A$ is $X = L^2(\Omega;\R^{3 \times 3})$. Accordingly, for $(\lambda, \operatorname{curl } \lambda) \in L^{\infty}(\Omega;\R^{3\times 3}) \times L^{2}(\Omega;\R^{3\times 3})$ the integral expression $\int_{\Omega} A:\operatorname{curl} \lambda  +  \lambda_{Zp} \eps_{pqr} \eps_{ZBC} A_{Bq} A_{Cr} - H(A) \, dx$ is well-defined. 
Similarly to the setting in Section \ref{sec:quad_H}, there exists a continuous, linear trace operator from $\{ \lambda \in L^2(\Omega;\R^{3\times 3}): \operatorname{curl } \lambda \in L^2(\Omega;\R^{3\times 3}) \}$ for the term $\lambda \times n$ into the space $W^{-1/2,2}(\partial \Omega;\R^{3\times 3})$ which is the dual of $W^{1/2,2}(\partial \Omega;\R^{3\times 3})$ (the trace space of $W^{1,2}(\Omega;\R^{3\times 3})$). Hence, the boundary term in \eqref{eq: def tildeS} has a well-defined interpretation as the dual pairing $\langle A^{(b)}, \lambda \times n \rangle_{W^{1/2,2},W^{-1/2,2}}$ if $A^{(b)} \in W^{1/2,2}(\partial \Omega;\R^{3\times})$ and it holds 
\begin{equation} \label{eq: boundary trace quadratic}
\left| \langle A^{(b)}, \lambda \times n \rangle_{W^{1/2,2},W^{-1/2,2}} \right| \leq C' \| A^{(b)} \|_{W^{1/2,2}} (\| \lambda\|_{L^{\infty}} + \| \operatorname{curl } \lambda \|_{L^2}).
\end{equation}
Again, in this setting it can be checked that $\tilde{S}_H$ is lower-semicontinuous with respect to $\lambda_k \stackrel{*}{\rightharpoonup} \lambda$ in $L^{\infty}(\Omega;\R^{3\times 3})$ and $\operatorname{curl} \lambda_k \rightharpoonup \lambda$ in $L^2(\Omega;\R^{3\times 3})$.
Hence, in the following we concentrate on showing coercivity of $\tilde{S}_H$ with respect to this convergence. Again, the key will be to show lower bounds for the function $g_{H}$.

\begin{proposition} \label{prop: lb quadratic}
Let $H(A) = \frac12 |A|^2$. Then there exists $c>0$ such that  
    \[
g_H(\lambda, \mu) = +\infty \text{ if } |\lambda| > \frac32 \text{ and }    g_H(\lambda, \mu) \geq c |\mu|^2 \text{ if } |\lambda| \leq \frac32.
    \]
    
\end{proposition}

\begin{proof}
First, let $|\lambda| > \frac32$ and consider $A \in \R^{3 \times 3}$ as in case 3 of the proof of Proposition \ref{prop: lb higher}, i.e.~let $\Lambda_V$ be a row of $\lambda$ such that $|\Lambda_V|^2 \geq \frac13 |\lambda|^2$ and $A \in \R^{3\times 3}$ such that $|A| = 2^{1/2} |\Lambda_V|^{1/2}$ and, as in (\ref{E:epqrZp}) with $\beta =2$,  
        \begin{equation*}
        \eps_{pqr} \eps_{ZBC} A_{Bq} A_{Cr}  = 
        \begin{cases} 2 \lambda_{Vp}  &\text{ if } Z=V, \\ 0 &\text{ otherwise.} 
        \end{cases}
        \end{equation*}

Then we obtain for $t>0$ and the matrix $\tilde{A} = tA$
\begin{align*}
    g_H(\lambda, \mu) &\geq -|\mu| |\tilde{A}| + \lambda_{Zp}  \eps_{pqr} \eps_{ZBC} A_{Bq} A_{Cr} - \frac12 |\tilde{A}|^2 \\
    &= - 2^{1/2} t |\mu| |\Lambda_V|^{1/2} + 2 t^2 |\Lambda_V|^2 - t^2 |\Lambda_V| \\
    &\geq -2^{1/2}t |\mu| |\lambda|^{1/2} + \frac23 t^2 |\lambda|^2 - t^2 |\lambda|  \stackrel{t \to \infty}{\longrightarrow} + \infty.
\end{align*}
Next, let $|\lambda| \leq \frac32$. 
Then we obtain using \eqref{eq: est K} for $A = \frac1{7} \mu$
\[
g_H(\lambda, \mu) \geq \mu:A - 2 |A|^2 |\lambda| - \frac12 |A|^2 \geq \mu:A - \frac{7}2 |A|^2 = \frac1{14} |\mu|^2.
\]
\end{proof}

Using this lower bound we show the coercivity of $\tilde{S}_H$.

\begin{proposition}\label{prop: coerc quadratic}
    For $H(A) = \frac12 |A|^2$ and $A^{(b)} \in W^{1/2,2}(\partial \Omega;\R^{3\times3})$ the dual functional $\tilde{S}_H$ is weakly(*) coercive on the Banach space $\{ \lambda \in L^{\infty}(\Omega; \R^{3\times 3}): \operatorname{curl}(\lambda) \in L^{2}(\Omega;\R^{3\times 3}) \}$, i.e. for all sequences $(\lambda_n, \operatorname{curl}(\lambda_n))_n \subseteq L^{\infty}(\Omega; \R^{3\times 3}) \times L^{2}(\Omega;\R^{3\times 3})$ such that $\sup_n \tilde{S}_H[\lambda_n] < \infty$ there exists a (not relabeled) subsequence and $\lambda \in L^{\infty}(\Omega;\R^{3\times 3})$ with $\operatorname{curl}(\lambda) \in L^{2}(\Omega;\R^{3\times 3})$ such that $\lambda_n \stackrel{*}{\rightharpoonup} \lambda$ in $L^{\infty}(\Omega;\R^{3\times 3})$ and $\operatorname{curl }(\lambda_n) \rightharpoonup \operatorname{curl }(\lambda)$ in $L^2(\Omega;\R^{2\times 2})$. 
\end{proposition}
\begin{proof}
    The result follows similarly to Proposition \ref{prop: coerc higher} using Proposition \ref{prop: lb quadratic} and \eqref{eq: boundary trace quadratic}.
\end{proof}

We close this section by briefly commenting on the case of a slightly modified quadratic function $\tilde{H}$ around a base state $\bar{A}$.

\begin{remark}
    Let $\bar{A} \in L^2(\Omega;\R^{3\times 3})$ and $\tilde{H}(x,A) = \frac12 |A - \bar{A}(x)|^2$. Then we can rewrite for every fixed $x \in \Omega$ and $\mu,\lambda \in \R^{3\times 3}$
    \begin{align*}
        g_{\tilde{H}(x,\cdot)}(\lambda, \mu) &:= \sup_{A \in \R^{3\times3}} \mu:A + \lambda_{Zp}  \eps_{pqr} \eps_{ZBC} A_{Bq} A_{Cr} - \tilde{H}(x,A) \\
        &= \sup_{A \in \R^{3\times3}} (\mu + \bar{A}(x)):A + \lambda_{Zp}  \eps_{pqr} \eps_{ZBC} A_{Bq}(x) A_{Cr} (x)   - H(A) - \frac12 |\bar{A}(x)|^2.
    \end{align*}
    Consequently, by the result above, we obtain 
        \[
g_{\tilde{H}(x,\cdot)}(\lambda, \mu) = +\infty \text{ if } |\lambda| > \frac32 \text{ and }    g_{\tilde{H}(x,\cdot)}(\mu,\lambda) \geq c |\mu + \bar{A}(x)|^2 - \frac12 |\bar{A}(x)|^2 \text{ if } |\lambda| \leq \frac32.
    \]
    It is straightforward to check that this estimate suffices to obtain coercivity of the dual functional for a modified quadratic function $\tilde{H}$ including a base state $\bar{A} \in L^2(\Omega;\R^{3\times 3})$. 
\end{remark}

\begin{appendix}

\section{Appendix: Lie-algebra valued forms}

We summarize here some notation and essential algebraic facts about differential forms with values in vector spaces and Lie algebras.
 
Let $A$ be a $k$-form on a manifold $M$, with values in a vector space $V$, and $B$ an $l$-form on $M$ with values in a vector space $W$. Then,
in terms of local coordinates $(x^1,\ldots, x^n)$, we define
\begin{equation}\label{E:AwedgeB]}
    A\wedge B \stackrel{\rm def}{=}  A_{i_1\ldots i_k}\otimes B_{j_1,\ldots, i_l} dx^{i_1}\wedge\ldots \wedge dx^{i_k}\wedge dx^{j_1}\wedge \ldots\wedge dx^{j_l},
\end{equation}
which is a $(k+l)$-form, with values in the vector space $V\otimes W$. (If $V$ and $W$ are equal, both being a vector space of matrices, it is more common to apply the matrix product to $V\otimes V$  and redefine $A\wedge B$ as a matrix-valued form.) This extends to wedge products of more forms in the natural way. If $\Phi:V \to Z$ is a linear mapping, where $Z$ is also a vector space, then, for a $V$-valued $k$-form $A$ as above,  we   define the $Z$-valued $k$-form $\Phi(A)$ by:
\begin{equation}\label{E:defPhiA}
\Phi(A)=\Phi(A_{i_1\ldots i_k}) dx^{i_1}\wedge\ldots \wedge dx^{i_k}.
\end{equation}
A bilinear map
$$\la\cdot,\cdot\ra:V\times V\to\R: (u, v)\mapsto \la u,v\ra$$
 induces a linear map 
 $$\la\cdot\ra: V\otimes V\to\R :u_i\otimes v_i\mapsto \la u_i, v_i\ra.$$
Now let $A$ be a $V$-valued $k$-form and $B$ a $V$-valued $l$-form; then, as in (\ref{E:AwedgeB]}),  $A\wedge B$ is a $V\otimes V$-valued $(k+l)$-form, and so, we can follow the recipe in (\ref{E:defPhiA}) to obtain
\begin{equation}\label{E:{AdotB]}}
    \la A\wedge B\ra_{k+l} \stackrel{\rm def}{=}\la A\wedge B\ra\stackrel{\rm def}{=} \la A_{i_1\ldots i_k},  B_{j_1,\ldots, i_l}\ra dx^{i_1}\wedge\ldots \wedge dx^{i_k}\wedge dx^{j_1}\wedge \ldots\wedge dx^{j_l},
\end{equation}
if $A$ and $B$ are both $V$-valued; note that, here $\la A\wedge B\ra $ is an ordinary $(k+l)$-form, not an inner-product of $A$ and $B$. {\em We will often drop the subscript $k+l$ in $\la\cdot\ra_{k+l}$ when it is clear from the context.}  If the bilinear form $\la\cdot,\cdot\ra$ is symmetric, then
     \begin{equation}\label{E:AwedgeBsym}
     \la A\wedge  B\ra  =(-1)^{kl}\la B\wedge A\ra.
   \end{equation}

Now let $\gog$ be a Lie algebra and 
\begin{equation}
    \gog\otimes\gog\to \gog: v\otimes w\mapsto [v,w]
\end{equation}
the linear map arising from the Lie bracket on $\gog$.  Then for $A$ a $\gog$-valued $k$-form and $B$ a $\gog$-valued $l$-form, we have the $\gog$-valued $(k+l)$-form
\begin{equation}\label{E:{AwedgeB]}}
    [A\wedge B] \stackrel{\rm def}{=} [A_{i_1\ldots i_k}, B_{j_1,\ldots, i_l}] dx^{i_1}\wedge\ldots \wedge dx^{i_k}\wedge dx^{j_1}\wedge \ldots\wedge dx^{j_l}.
\end{equation}
Assume that $\la\cdot,\cdot\ra$ is a symmetric, bilinar form on $\gog$. We assume, further, that the following {\em ad-invariance} relations hold:
\begin{equation}\label{E:uvw{}}
    \la [u,v], w\ra = -\la v, [u,w]\ra,
\end{equation}
for all $u, v,w\in \gog$.  If $\gog$ is a matrix Lie algebra, then the pairing $(u,v)\mapsto {\rm Tr}(uv)$ satisfies this condition, as is seen from (\ref{E:TrXYZ}). Symmetry of $\la\cdot,\cdot\ra$ and the relation (\ref{E:uvw{}}) imply the cyclic relations:
\begin{equation}
   \la [u,v], w\ra =  \la [w,u],v\ra=\la [v,w], u\ra
\end{equation}

\begin{lemma}
   Let $A$ be a $k-$form, $B$ an $m$-form, and $C$ and $n$-form, all defined on a manifold, with values in a Lie algebra $\gog$, which is equipped with a symmetric, bilinear pairing $\la\cdot,\cdot\ra$ on $\gog\times\gog$, satisfying the ad-invariance condition (\ref{E:uvw{}}). Then
   \begin{equation}\label{E:AbrackBskew}
       [A\wedge B]=-(-1)^{km}[B,A]
   \end{equation}
   and 
   \begin{equation}\label{E:ABwC}
       \la A\wedge [B\wedge C]\ra_{k+m+n} =(-1)^{(k+m)n}\la C\wedge [A\wedge B]\ra_{k+m+n}=(-1)^{(n+m)k}\la B\wedge [C\wedge A]\ra_{k+m+n}.
   \end{equation}
   If $\gog$ is a matrix Lie algebra, then, for any $\gog$-valued $1$-form $A$, we have
   \begin{equation}\label{E:AwedgeA}
       [A\wedge A] =2 \Pi(A\wedge A),
   \end{equation}
   where $\Pi:\gog\otimes\gog\to\gog:u\otimes v\mapsto uv$, the matrix product.
\end{lemma}
As a rule, for matrix-valued forms, we drop $\Pi$ and simply write  $A\wedge A$ and $A\wedge A\wedge A$. The utility of (\ref{E:ABwC}) is that it allows us to move each of the terms $A$, $B$, $C$ to the left. 
\begin{proof}
    Equations  (\ref{E:AbrackBskew}) and  (\ref{E:ABwC}) are readily verified.
    Next, for any $\gog$-valued $1$-form $A$, we have
    \begin{equation}
    \begin{split}
          [A\wedge A] &=[A_i,A_j]dx^i\wedge dx^j\\
          &= A_iA_jdx^i\wedge dx^j-A_jA_idx^i\wedge dx^j\\
          &= A_iA_jdx^i\wedge dx^j+A_jA_idx^j\wedge dx^i\\
          &=2A_iA_jdx^i\wedge dx^j\\
          &=2A\wedge A,
    \end{split}
    \end{equation}
    where we have dropped $\Pi$.
  \end{proof}

The {\em curvature} form $F^A$, associated to $A$, is defined by
\begin{equation}\label{E:defFA}
    F^A=A+\frac{1}{2}[A\wedge A].
\end{equation}
When $A$ is matrix-valued, we can then write this as
\begin{equation}\label{E:defFAmatr}
    F^A=A+ A\wedge A.
\end{equation}

\subsection{The group $SU(2)$}\label{ss:su2}

Some of our results are worked out for the group $SU(2)$. Here we summarize some basic facts and notation about $SU(2)$ and its Lie algebra $su(2)$.

The matrix group $SU(2)$ is, explicitly, given by
\begin{equation}
    SU(2)=\left\{\left(\begin{matrix} x_1+ix_2 & - x_3+ix_4\\ x_3+ix_4 &  x_1-ix_2\end{matrix}\right) :\, x_1^2+x_2^2+x_3^2+x_4^2=1\right\}. 
\end{equation}
Thus, $SU(2)$ is essentially the unit sphere in $\R^4$, with the identity matrix $I$ corresponding to the point $(1,0,0,0)$. The Lie algebra $su(2)$ is the tangent space $T_ISU(2)$ and is thus
\begin{equation}
    su(2)=T_ISU(2)=\left\{i \left[\begin{matrix}y_3 & y_1 -iy_2\\ y_1+iy_2 & - y_3 \end{matrix}\right]\,: y_1, y_2, y_3\in \R^3\right\}.
\end{equation}
A basis of $su(2)$ is given by $i\sigma_1, i\sigma_2, i\sigma_3$, where we are using the Pauli matrices
\begin{equation}\label{E:Pauli}
\sigma_1=\left(\begin{matrix} 0 & 1 \\ 1& 0\end{matrix}\right),\quad \sigma_2=\left(\begin{matrix} 0 & -i \\ i & 0\end{matrix}\right),
\quad \sigma_3=\left(\begin{matrix} 1& 0\\ 0& -1\end{matrix}\right).
\end{equation}
Now consider the  inner-product on complex matrices given by
\begin{equation}\label{E:ipmatr}
    \la A, B\ra =\frac{1}{N} {\rm Tr}(A^*B)=\frac{1}{N}\sum_{j,k=1}^N A_{jk}^*B_{jk},
\end{equation}
the scaling by $1/N$ is to normalize $\la I, I\ra$ to have the value $1$. The  matrices 
$$E_1=i\sigma_1, \, E_2= i\sigma_2,\, E_3= i\sigma_3$$
form a  basis of $su(2)$. We note also that since each element of $su(2)$ is skew-hermitian, the inner-product (\ref{E:ipmatr}) reduces on $su(2)$ to
\begin{equation}\label{E:ipsu2}
     \la A, B\ra =-\frac{1}{2}{\rm Tr}(AB). 
\end{equation}
There are the standard identities:
\begin{equation}\label{E:EJKid}
    \begin{split}
        E_1^2 = E_2^2&=E_3^2   =-I\\
        E_1E_2=-E_2E_1=-E_3,\,\,&  E_3E_1=-E_1E_3= -E_2,\, E_2E_3=-E_3E_2= -E_1, 
    \end{split}
\end{equation}
which imply the commutation relations:
\begin{equation}\label{E:EJKLcomm}
    [E_1, E_2]=-2E_3,\,\, [E_3, E_1]= -2E_2, \,\, [E_2, E_3]=-2E_1.
\end{equation}
Orthonormality of the basis  $\{E_J\}$  of $su(2)$ arises from the relations:
\begin{equation}\label{E:TrEJK}
    -{\rm Tr}(E_JE_K)= 2\delta_{JK}.
\end{equation}
We have then
\begin{equation}\label{E:TrEJKL}
    -{\rm Tr}\bigl(E_J[E_K, E_L]\bigr)=-{\rm Tr}\bigl(E_L[E_J, E_K]\bigr)=-{\rm Tr}\bigl(E_K[E_L, E_J]\bigr)=4\epsilon_{JKL}.
\end{equation}
In fact, for any  square matrices $X, Y, Z$, all of the same size, we have
\begin{equation}\label{E:TrXYZ}
    \Tr(X[Y,Z])=\Tr(Z[X,Y])=\Tr(Y[Z,X]).
\end{equation}
We note that, by definition of the inner-product (\ref{E:ipsu2}), equation (\ref{E:TrEJKL}) implies:
\begin{equation}\label{E:TripEJKL}
   \la [E_J, E_K], E_L]\ra = 2\epsilon_{JKL}.
\end{equation}

\section{Appendix: Function space calculus}
We summarize some notions and results on Fr\'echet spaces, using the work of Hamilton \cite{Ham1982}.

A {\em Fr\'echet space} is a real or complex vector space $X$, equipped with a topology arising from a sequence of seminorms $\{|\!|\cdot|\!|_n\}_{n\geq 0}$, such that the topology is Hausforff and complete (every Cauchy sequence converges). As an example (\cite[Example 1.1.5]{Ham1982}), for a vector bundle $E$ over a compact Riemannian manifold $M$,   the vector space $C^\infty(M,V)$  of all $C^\infty$ sections $M\to V$ is a Fr\'echet space with respect to the topology induced by  the family of  norms $|\!|\cdot|\!|_n$   given by
$$|\!|f|\!|_n=\sum_{k=0}^n\sup_{p\in M}|D^kf(p)|,$$
where $D^kf$ is the $k$-th covariant derivative of $f$. The larger space $C^k(M,V)$, for finite $k$, is a Banach space with the norm $|\!|\cdot|\!|_k$.

In the sequel, $X$, $Y$, and $Z$ are Fr\'echet spaces.

\begin{definition}\label{D:direcder}
For a function $F:U\to Y$, where $U$ is an open subset of $X$, and both $X$ and $Y$ are Fr\'echet spaces, the {\em directional derivative} $F'(x)h$ or $DF|_xh$ is defined by
\begin{equation}
    F'(x)h=DF|_xh=\lim_{t\to 0}\frac{1}{t}\left[F(x+th)-F(x)\right].
\end{equation}
We will say that $F$ is {\em directionally $C^1$} on $U$ if $F'(x)h$ exists for all $(x,h)\in U\times X$, and $(x,h)\mapsto F'(x)h$ is continuous.  
\end{definition}

According to \cite[Theorem 3.2.5]{Ham1982}, if $F$ is directionally $C^1$ then $h\mapsto F'(x)h$ is linear for all $x\in U$.

By  \cite[Lemma 3.3.1]{Ham1982}, a map $F:U\to Y$, where $U$ is a convex open subset of $X$, is directionally $C^1$ if and only if there is a continuous map $L:U\times U\times X\to Y$, with $h\mapsto L(f_0,f_1)h$ being linear for all $f_0, f_1\in U$  such that
\begin{equation}
    F(f_1)-F(f_0)= L(f_0, f_1)(f_1-f_0).
\end{equation}
If this holds, then
\begin{equation}
    F'(f)h=L(f,f)h 
\end{equation}
for all $f\in U$ and $h\in X$. The proof is obtained by writing $F(f_1)-F(f_0)$  as the integral $\int_0^1\ldots dt$ of the derivative of the function $[0,1]\to Y:t\mapsto F'((1-t)f_0+tf_1)h$. From this follows Hamilton's Fr\'echet space {\em chain rule}\cite[Theorem 3.3.4]{Ham1982}:
\begin{equation}\label{E:HamCR}
    D(G\circ F)|_fh =DG|_{F(f)}DF|_fh
\end{equation}
whenever $F:U\to Y$ and $G:V\to Z $ are composable Fr\'echet $C^1$ functions, where $U$ is an open subset of $X$ and $V$ is an open subset of $Y$m with $f\in U$ and $h\in X$.
\section{Appendix: Coercivity for Quartic $H$}\label{app:quartic}

  Consider the quartic function $H(A) =  \ell |A|^4$ with $\ell=\left(\frac{3}{8}\right)^6 \frac1{4 \cdot 18}$.
 
  In this setting, the natural space for the primal variables $A$ is $X = L^4(\Omega;\R^{3\times 3})$ and, accordingly, $(\lambda,\curl \lambda) \in L^2(\Omega; \R^{3\times 3}) \times L^{4/3}(\Omega;\R^{3\times3})$, where $\operatorname{curl} \lambda$ has to be understood in the sense of distributions. Note that these choices guarantee by H\"older's inequality that the integral 
  $$\int_{\Omega} A:\operatorname{curl} \lambda  +  \lambda_{Zp} \eps_{pqr} \eps_{ZBC} A_{Bq} A_{Cr} - H(A) \, dx$$
  is well-defined. Next, note that for functions  $\lambda \in L^2(\Omega; \R^{3\times 3}) \subseteq L^{4/3}(\Omega; \R^{3\times 3})$ with $\operatorname{curl} \lambda \in L^{4/3}(\Omega; \R^{3\times 3})$ there exists a linear, continuous trace operator for the term $\lambda \times n$ into the space $W^{-3/4,4}(\partial \Omega; \R^{3\times 3})$, the dual of the space $W^{3/4,4}(\partial \Omega; \R^{3\times 3})$ which in turn is the trace space of $W^{1,4}(\Omega,\R^{3\times 3})$, c.f.~see \cite[Theorem 18.40]{LeoniSobolev} and \cite[Theorem 1, page 204]{Lions}. In particular, assuming that $A^{(b)} \in W^{3/4,4}(\partial \Omega;\R^{3\times 3})$ (i.e., $A^{(b)}$ is the trace of a function in $W^{1,4}(\Omega;\R^{3\times 3})$) then the well-defined interpretation of the integral $\int_{\partial \Omega} A^{(b)} : (\lambda \times n) \, da$ is the dual pairing $\langle A^{(b)}, \lambda \times n \rangle_{W^{3/4,4},W^{-3/4,4}}$. This expression is continuous with respect to the weak convergences $\lambda_k \rightharpoonup \lambda$ in $L^4$ and $\operatorname{curl} \lambda_k \rightharpoonup \operatorname{curl} \lambda$ in $L^{4/3}$. Moreover, it holds 
  \begin{equation}\label{eq:quart_trace}
      \left| \langle A^{(b)}, \lambda \times n \rangle_{W^{3/4,4},W^{-3/4,4}} \right| \leq \|A^{(b)} \|_{W^{3/4,4}} \| \lambda \times n\|_{W^{-3/4,4}} \leq C' \|A^{(b)} \|_{W^{3/4,4}} \left( \| \lambda\|_{L^2} + \| \operatorname{curl } \lambda \|_{L^{4/3}} \right).
  \end{equation} 

  It can then be shown that in this setting $\tilde{S}_H$ is lower-semicontinuous with respect to $\lambda_k \rightharpoonup \lambda$ in $L^4$ and $\operatorname{curl} \lambda_k \to \operatorname{curl} \lambda$ in $L^{4/3}$.
  Hence, the main step to prove existence of a minimizer for $\tilde{S}_H$ using the Direct Method of the Calculus of Variations is to prove coercivity of $\tilde{S}_H$.
  In order to show the coercivity of $\tilde{S}_H$ we first prove lower bounds on the function $g_H$. 

  \begin{proposition}\label{prop: lb quartic}
      Let $H(A) = \left(\frac{3}{8}\right)^6 \frac1{4 \cdot 18} |A|^4$. Then there exists $c>0$ such that 
      \[
      g_H(\lambda, \mu) \geq c \left( |\mu|^{4/3} + |\lambda|^2 \right).
      \]
  \end{proposition}

  \begin{proof}
      Fix $\lambda, \mu \in \R^{3\times 3}$.  
      We distinguish the following cases.
      \begin{enumerate}
          \item Assume that $|\lambda| \leq \frac38 |\mu|^{2/3}$. 
          Let $A = \mu |\mu|^{-2/3}$.
          Then using \eqref{eq: est K} it follows
          \begin{align*}
              g_H(\lambda, \mu) &\geq \mu:A - 2 |\lambda| |A|^2 - \frac18|A|^4 = \frac78 |\mu|^{4/3} - 2 |\lambda| |\mu|^{2/3} \geq \frac18 |\mu|^{4/3} \geq \frac1{16} \left( |\mu|^{4/3} + |\lambda|^2 \right). 
          \end{align*}
          \item Assume that $\frac1{12} |\lambda|^{3/2} \leq |\mu| \leq \left( \frac83 \right)^{3/2} |\lambda|^{3/2}$. 
          For $A = (4|\lambda|)^{-1} \mu$ we find
          \begin{align*}
              g_H(\lambda, \mu) &\geq \mu: A - 2 |\lambda| |A|^2 - \left(\frac{3}{8}\right)^6 \frac1{4 \cdot 18} |A|^4 \\
              &= \frac14 |\lambda|^{-1} |\mu|^2 - \frac1{8} |\lambda|^{-1} |\mu|^2 - \left(\frac{3}{8}\right)^6 \frac1{4 \cdot 18 \cdot 256} |\lambda|^{-4} |\mu|^4 \\
              &\geq \frac1{8 \cdot 12^2} |\lambda|^2 - \left(\frac{3}{8}\right)^6 \frac1{4 \cdot 18 \cdot 256} \left( \frac83 \right)^6  |\lambda|^2   \\
              &\geq \frac1{16 \cdot 12^2}|\lambda|^2 \geq \frac1{32 \cdot 12^2} \left(\frac38 \right)^{2} \left( |\lambda|^2 + |\mu|^{4/3} \right).
          \end{align*}
            
           \item Assume that $|\mu| \leq \frac1{12} |\lambda|^{3/2}$. Let $V \in \{1,2,3\}$ such that it holds for the $V^{th}$ row, $\Lambda_V \in \R^3$, of $\lambda \in \R^{3\times 3}$ that $ m^2:= \lambda_{Vi} \lambda_{Vi} \mbox{ (no sum on $V$) } = |\Lambda_V|^2  \geq \frac13 |\lambda|^2$, $m > 0$. 
            Next, let $Y,Z \in \{1,2,3\}$ such that $\{V,Y,Z \} = \{1,2,3\}$. 
          Choose an orthonormal basis $(E_1, E_2, E_3)$ with $E_V = \frac{1}{m}\Lambda_V$ and $\half \eps_{VYZ} E_Y \times E_Z = E_V$.
          Now, define $A \in \R^{3\times 3}$ row-wise as 
          \[
          A_{Yi} =  \sqrt{m} \, E_{Yi} \quad \text{ if } Y \neq V \qquad \text{ and } \qquad A_{Vi} = 0.
          \]
  It follows that $|A| = \sqrt{2m}$ and 
          \begin{equation}\label{E:epqrZp quartic}
          \eps_{pqr} \eps_{ZBC} A_{Bq} A_{Cr}  = 
          \begin{cases} 2 m \, E_{Vp} = 2 \lambda_{Vp}  &\text{ if } Z=V, \\ 0 &\text{ otherwise.} 
          \end{cases}
          \end{equation}
        
          Then we estimate
          \begin{align*}
              g_H(\lambda, \mu) &\geq -|\mu| |A| + 2 |\Lambda_V|^2 - \frac1{4} |A|^4 \geq - 2^{1/2} |\mu| |\Lambda_V|^{1/2} + |\Lambda_V|^2 \\
              &\geq -2 |\mu| | \lambda|^{1/2} + \frac13 |\lambda|^2 
              \geq \frac16 |\lambda|^2  \geq \frac1{12} \left( |\lambda|^2 + |\mu|^{4/3} \right).
          \end{align*}
      \end{enumerate}
  \end{proof}

  \begin{proposition}\label{prop: coerc quartic}
      For $H(A) = \left(\frac{3}{8}\right)^6 \frac1{4 \cdot 18} |A|^4$ and $A^{(b)} \in W^{3/4,4}(\partial \Omega;\R^{3\times3})$ the dual functional $\tilde{S}_H$ is weakly coercive on the Banach space $\mathcal{A}:=\{ \lambda \in L^2(\Omega; \R^{3\times 3}): \operatorname{curl}(\lambda) \in L^{4/3}(\Omega;\R^{3\times 3}) \}$, i.e., for every sequence $(\lambda_k)_k \subseteq \mathcal{A}$ such that $\sup_k \tilde{S}_H[\lambda_k] < \infty$ there exists a (not relabeled) subsequence and $\lambda \in \mathcal{A}$ such that $\lambda_k \rightharpoonup \lambda$ in $L^2$ and $\operatorname{curl} (\lambda_k)\rightharpoonup \operatorname{curl}(\lambda)$ in $L^{4/3}$.
  \end{proposition}
  \begin{proof}
      By Proposition \ref{prop: lb quartic} and the trace estimate discussed above Proposition \ref{prop: lb quartic} for a sequence $(\lambda_k)_k \subseteq \mathcal{A}$ satisfying $\sup_k \tilde{S}_H[\lambda_k] \leq C < \infty$ it follows for all $k\in \N$ that 
      \[
       c ( \| \lambda_k \|_{L^2}^2 + \| \operatorname{curl } \lambda_k \|_{L^{4/3}}^{4/3} ) - C'  \|A^{(b)} \|_{W^{3/4,4}} \left( \| \lambda_k \|_{L^2} + \| \operatorname{curl } \lambda_k \|_{L^{4/3}} \right) \leq \tilde{S}_H[\lambda_k] \leq  C.
      \]
      Since $\lim_{\alpha, \beta \to +\infty} c ( \alpha^2 + \beta^{4/3} ) - C'  \|A^{(b)} \|_{W^{3/4,4}} \left( \alpha + \beta \right) = + \infty$,
   it follows that the terms $\| \lambda_k\|_{L^2}$ and $ \| \operatorname{curl}(\lambda_k)\|_{L^{4/3}}$ are uniformly bounded. By the usual compactness arguments this implies the existence of a (not relabeled) subsequence and functions $\lambda \in L^2(\Omega;\R^{3\times 3})$ and $\mu \in L^{4/3}(\Omega;\R^{3\times 3})$ such that $\lambda_n \rightharpoonup \lambda$ in $L^2$ and $\operatorname{curl} (\lambda_n)\rightharpoonup \mu$ in $L^{4/3}$. Eventually, it can then be checked that actually $\operatorname{curl}(\lambda) = \mu$. 
  \end{proof}

  Note that by similar arguments one can show directly from Proposition \ref{prop: lb quartic} and \eqref{eq:quart_trace} that the dual functional $\tilde{S}_H$ is bounded from below.
  
\end{appendix}
\section*{Acknowledgments}
The work of AA and ANS was supported by the Simons Pivot Fellowship grant \# 983171.

\printbibliography

\end{document}